\documentclass[a4paper]{article}

\usepackage[a4paper, total={6in, 10in}]{geometry}

\usepackage{authblk}

\title{Choreographies for Reactive Programming}

\author[1]{Marco Carbone} 
\author[2]{Fabrizio Montesi} 
\author[3]{Hugo {Torres Vieira}} 

\affil[1]{IT University of Copenhagen, Copenhagen, Denmark}
\affil[2]{University of Southern Denmark, Odense, Denmark}
\affil[3]{IMT School for Advanced Studies Lucca, Lucca, Italy}
\date{}

\usepackage{amsmath}
\usepackage{amssymb}
\usepackage{stmaryrd}
\usepackage[usenames,dvipsnames]{xcolor}
\usepackage{proof-dashed}
\usepackage{tikz}
\usetikzlibrary{arrows}
\usetikzlibrary{positioning}
\usepackage{todonotes}
\usepackage{hyperref}
\usepackage[capitalise]{cleveref}
\usepackage{times}

\usepackage{amsthm}

\newtheorem{example}{Example}[section]
\newtheorem{proposition}{Proposition}[section]
\newtheorem{theorem}{Theorem}[section]
\newtheorem{lemma}{Lemma}[section]
\newtheorem{corollary}{Corollary}[section]

\newcommand{\m}[1]{\mathsf{#1}}

\newcommand{\pp}{\,|\,}

\newcommand{\deff}{\triangleq}

\newcommand{\parop}{\,|\,}

\newcommand{\chin}{\&}

\newcommand{\buildqueues}[2]{\mathit{queues}(#1,#2)}
\newcommand{\queues}[3]{\mathit{fill}(#1,#2,#3)}
\newcommand{\lbox}[1]{[#1]}

\newcommand{\inl}{{\mathtt{inl}}}
\newcommand{\inr}{{\mathtt{inr}}}

\newcommand{\recx}[1]{\recvar{#1}}

\newcommand{\zero}{\mathbf{0}}

\newcommand{\msgout}[1]{{!}(#1)}
\newcommand{\msgin}[1]{{?}(#1)}

\newcommand{\projc}[4]{#1\! \downarrow_{#2,#3,#4}}
\newcommand{\projb}[1]{[\![ #1 ]\!]}

\newcommand{\nub}{{\bf \nu }}

\newcommand{\lto}[1]{\ \mathrel{\stackrel{{\;\;#1\;\;}}{\mbox{\rightarrowfill}}}\ }

\newcommand{\ltow}[1]{\mathrel{\stackrel{{\;\;#1\;\;}}{\Longrightarrow}}}
\newcommand{\lfrom}[1]{\mathrel{\stackrel{{\;\;#1\;\;}}{\mbox{\leftarrowfill}}}}

\newcommand{\subst}[2]{\{{}^{#1}\! / \!{}_{#2}\}}

\newcommand{\role}[1]{\mathsf{#1}}

\newcommand{\gcom}[3]{\role{#1}\!\!\!\lto{#2}\!\!\!\role{#3}}
\newcommand{\gchoice}[5]{\role{#1}\!\!\!\lto{#2}\!\!\!\role{#3}(#4,#5)}
\newcommand{\recvar}[1]{\mathbf{#1}}
\newcommand{\gend}{\mathbf{end}}

\newcommand{\lab}{\ell}

\newcommand{\roleport}[2]{\role{#1}.{#2}}

\newcommand{\dbinder}[4] {   \role{#2}.{#3}\lfrom{#1} {#4} }
\newcommand{\lbinder}[3] {   {#1} \,=\, {#3}\,\langle\,{#2} }
\newcommand{\lbinderp}[2] {   {#1} = {#2}}
\newcommand{\fbinder}[2]{ #1\leftarrow #2 }

\newcommand{\roleas}[2]{\role{#1}\!=\!#2}

\newcommand{\interface}[2]{[{#1}\,\rangle\,{#2}]}
\newcommand{\chorboxb}[3]{\interface{#1}{#2} \{ #3 \}}
\newcommand{\chorbox}[6]{\interface{#1}{#2} \{ #6; #5; #3; #4\}}

\newcommand{\labout}[3]{#1!#2\langle#3\rangle}
\newcommand{\labinp}[3]{#1?#2\langle#3\rangle}

\newcommand{\fbinderl}[2]{ #1 \ll #2 }
\newcommand{\fbinderr}[2]{ #1 \gg #2 }

\newcommand{\mapping}{\gamma}

\newcommand{\tout}[3]{#1!#2. #3}
\newcommand{\tinp}[3]{#1?#2. #3}
\newcommand{\tchoice}[3]{#1 \oplus (#2, #3)}
\newcommand{\tbranch}[3]{#1\, \& (#2, #3)}
\newcommand{\tend}{\gend}

\newcommand{\chot}{{\mathtt{ChoT}}}

\newcommand{\conforms}{\asymp}
\newcommand{\join}{\Join} 
\newcommand{\rename}[4]{#2_{#3,#4}(#1)}
\newcommand{\subt}{{\;<:\;}}

\newcommand{\recenv}{\Gamma}
\newcommand{\labenv}{\Delta}

\newcommand{\absL}{{\mathcal{L}}}
\newcommand{\abs}[1]{\mathit{abs}(#1)}

\newcommand{\absv}{\eta}
\newcommand{\til}{\tilde}

\newcommand{\eval}{\downarrow}
\newcommand{\rname}[1]{\ensuremath{\m{#1}}}
\newcommand{\dom}{\m{dom}}

\begin{document}

\maketitle
\begin{abstract}
  Modular programming is a cornerstone in software development,
  as it allows to build complex systems from the assembly of simpler
  components, and support reusability and substitution principles. In a
  distributed setting, component assembly is supported by
  communication that is often required to follow a prescribed protocol
  of interaction. In this paper, we present a language for the modular
  development of distributed systems, where the assembly of components
  is supported by a choreography that specifies the communication
  protocol. Our language allows to separate component behaviour, given
  in terms of reactive data ports, and choreographies, specified as
  first class entities. This allows us to consider reusability and
  substitution principles for both components and choreographies. We
  show how our model can be compiled into a more operational
  perspective in a provably-correct way, and we present a typing
  discipline that addresses communication safety and progress of
  systems, where a notion of substitutability naturally arises. 
\end{abstract}

\section{Introduction}
\label{sec:introduction}
In Component-Based Software Engineering (CBSE), software is built by
composing loosely-coupled components. The hallmark of CBSE is
reusability: the same component can be taken ``off the shelf'' and
reused in many different systems, as long as it is used accordingly to
its interface \cite{M68}.

Recently, CBSE is experiencing a renaissance. One reason is that it
adapts well to the complexity of modern computing paradigms, like
cloud computing, where building software whose components can be
deployed on separate computers is an advantage. The extreme of this
method is the emerging development paradigm of microservices, where
all components are autonomous and reusable services (the
microservices) that communicate through message passing
\cite{DGLMMMS16}.  Another reason is that programming techniques such
as reactive programming, where computation is performed in response to
newly available data, are becoming mainstream and make easier the
development of responsive components (e.g., as in services or
graphical user interfaces).

In settings where components exchange messages, e.g., microservices,
communications are expected to follow some predefined
protocols. Protocols are typically expressed in terms of some
choreography, an ``Alice and Bob'' description that prescribes which
communications should take place among participants, and in which
order. Choreography specifications can be found in the Web Services
Choreography Description Language by the W3C \cite{wscdl} and in the
Business Process Modelling Notation by the Object Management Group
\cite{bpmn}. Message Sequence Charts can also be seen as early
choreography models \cite{MSC}.

To implement its role in a protocol, a component must implement a
certain sequence of I/O actions.  A direct way of achieving this is to
use a language that natively supports sequencing I/O, like a workflow
language (e.g., the Business Process Execution Language \cite{BPEL})
or languages with actor- or process calculus-like constructs (e.g.,
Erlang \cite{AVW93} or Jolie \cite{MGZ14}).  Then, we can check that
the sequence of I/O actions for each protocol used by a program
follows the correct order \cite{HYC16}.

Unfortunately, this way of proceeding hinders reusability, in that a
component can then be used only in environments that accept exactly
the sequence of I/O performed by it. This makes the whole method
depend on how future protocols are designed.  For example, assume
that we design a protocol where a component receives a number and then
must output its square root exactly once. What if, tomorrow, we need
to use this component in an context that requires computing a
square root twice?

Reactive programming addresses this kind of problems by defining behaviour in response to external stimulus. For example, we can design a component that outputs the square root of a received number whenever requested. Thus reusability is better.
However, here we lose the clear connection to the notion of protocol, since reactive components can be ``too wild''. What if the designer actually cares about the fact that the component for computing the square root is invoked exactly once or twice (maybe for resource reasons)? Enforcing this kind of constraints usually requires obscure bookkeeping or side-effects, making programming error-prone and verification challenging.

In this paper, we attempt at having the best of both worlds. We propose a new 
language model, called Governed Components (GC), where reactive programming is married to 
choreographic specifications of communication protocols.
We believe that our results represent a first fundamental step towards merging the flexibility 
of reactive programming with the necessity of generating provably-correct I/O behaviour in 
communicating systems.

In GC, the computation performed by components is defined in the reactive style, using 
binders that dynamically produce results as soon as they get the input data that they need.
The key novelty is then in how components can be composed: a composition of components is 
always associated to a protocol, given as a choreography, that governs the flow of communications 
among the components. This means that, among all the possible reactions supported by the 
composed components, only those that are allowed by the protocol are actually executed.
The composition of some components is itself a component which can be used in further compositions.
We study the applicability of GC by formally defining a compiler from GC to a
model of concurrent processes with standard I/O primitives. The compiler translates the protocols 
used in compositions to a distributed process implementation, illustrating how GC can be used in 
practice. We prove that the compiler preserves the intended semantics of components, through an 
operational correspondence result.

Thanks to the marriage of reactive programming and protocol specifications, GC supports 
a new interesting substitution principle that is not supported by previous models based on 
processes (like \cite{HYC16}): a component may be used in compositions under different protocols, 
as long as the reactions supported by the component are enough to implement its part in the 
protocols.
This principle improves reusability in different directions.
One example is abstraction from message ordering: if a component needs two values to perform a 
computation, it can be used with any protocol that provides them without caring about the order in 
which the protocol will make them arrive.
Another example is abstraction from the number of reactions: if a component can perform a 
computation in reaction to some data inputs, then it can be used with different protocols that 
require such computation different numbers of times (e.g., zero, exactly two, or an unbounded 
number of times).
We first present our model and informally illustrate the semantics with a series of examples. Then, we define a typing discipline that is sound with respect to the substitution principle. In other words, typing 
guarantees that each component provides \emph{at least} enough reactive behaviour as needed by the 
protocol that it participates in, which we use to prove that well-typed component systems enjoy 
progress (never get stuck). From our operational correspondence result, it follows that also 
concurrent processes compiled from well-typed GC programs never get stuck.
We end our development by presenting a preliminary formal investigation of a sound subtyping relation for our typing, which captures part of the sound substitutions in our model.
\section{Language Preview}
\label{sec:preview}
We informally introduce our language by revisiting the Buyer-Seller-Shipper
example (BSS)~\cite{HYC08,HYC16}. This example consists of a 
$\role{Buyer}$ sending a purchase request for some item to a 
$\role{Seller}$. $\role{Seller}$ reacts by sending back the price of
the item. Based on the received price, $\role{Buyer}$ either accepts or
rejects the offer. In the first case, both $\role{Seller}$ and a
third-party delivery service $\role{Shipper}$ are notified by
$\role{Buyer}$ which will end the transaction by sending its credit
card details to $\role{Seller}$ and its delivery address to
$\role{Shipper}$. If $\role{Buyer}$ rejects the offer, the protocol
ends.

Choreographically, the behaviour of the BSS protocol can be written as
follows:
\vspace{-6pt}
\begin{displaymath}
  G_{\text{BSS}}=
  \begin{array}{llll}
    & \gcom{Buyer}{prod}{Seller};\\
    & \gcom{Seller}{price}{Buyer};\\
    & \gcom{Buyer} {decision}{Seller, Shipper}
      \left(
      \begin{array}{ll}
        \gcom {Buyer}{cc}{Seller};\gcom {Buyer}{dst}{Shipper},\;
        \tend
      \end{array}
      \right)
      \vspace{-6pt}
  \end{array}
\end{displaymath}
The protocol above describes the expected interactions between roles
$\role{Buyer}$, $\role{Seller}$ and $\role{Shipper}$. For example,
$\gcom{Buyer}{prod}{Seller}$ says that $\role{Buyer}$ must send a
message labeled as $prod$ to $\role{Seller}$. In the last line,
$\role{Buyer}$ can take the decision whether to accept or not the
purchase. Here, the protocol branches to two different subprotocols:
it either proceeds with the payment or terminates.

If we were to implement a system ruled by the protocol above, we may
use off-the-shelf modules implementing the various roles, provided
that they comply with the BSS protocol. For example, $\role{Buyer}$
could be implemented by the following component $C_{\role{Buyer}}$,
where the $\role{Buyer}$ role annotation is a mere graphical friendly annotation 
since our components are anonymous for the purpose of reuse:
\vspace{-6pt}
\begin{center}
  \includegraphics[scale=.35]{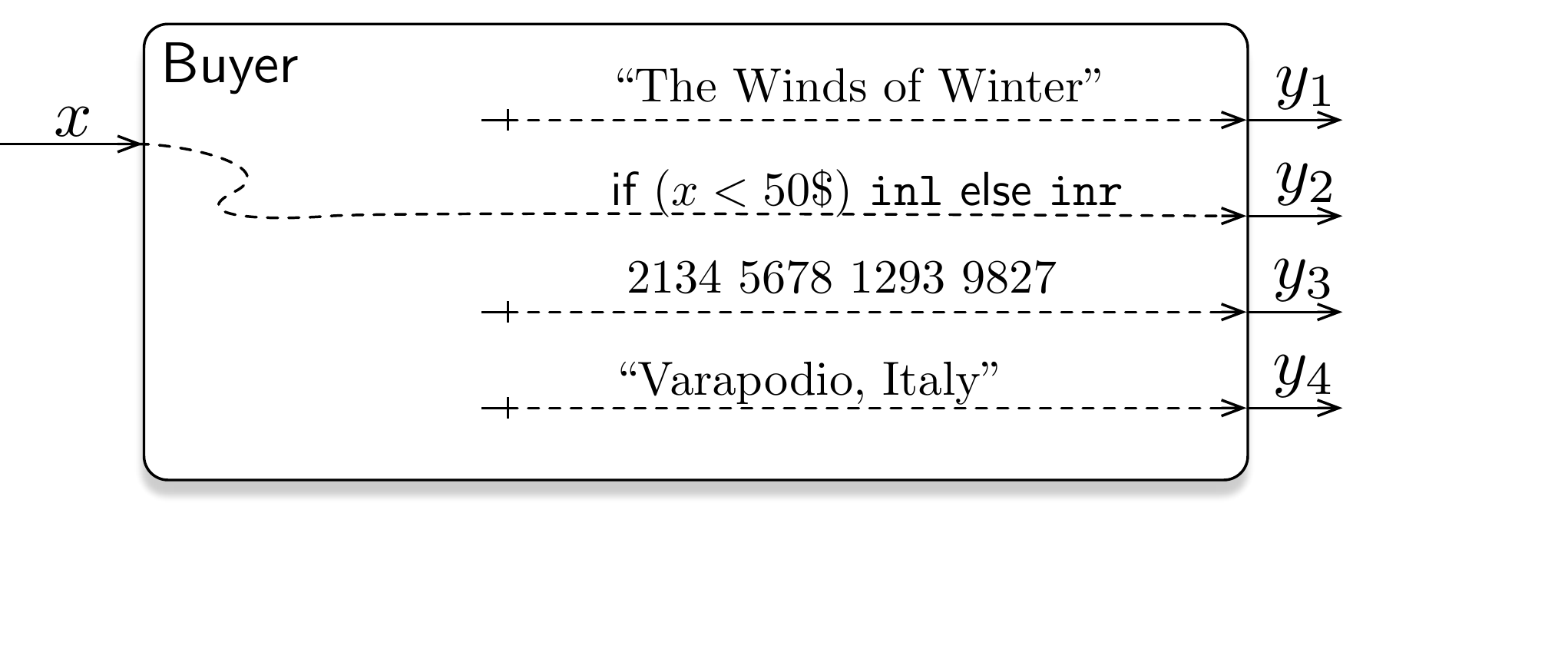}
\end{center}
\vspace{-32pt}
The one above is an informal but intuitive graphic representation of
the formal syntax which will be introduced in the next section. 
The component has an {\em interface}, denoted by
$\interface {x}{y_1, y_2,y_3,y_4}$, specifying that it can always
receive values on the {\em input port} $x$ and may output values on
{\em output ports} $y_1$, $y_2$, $y_3$, and $y_4$. {\em Local binders}
(dashed lines) define how such values are provided. For example, the local binder
$ \lbinderp{y_4} {\text{``Varapodio, Italy''}}$ says that output
variable $y_4$ is constantly able to output the value (address)
``Varapodio, Italy'' and $\lbinderp{y_2}{\textsf{if }(x<50)\  \inl \textsf{ else } \inr}$ is able to output a
choice decision based on the value received on input variable $x$,
where $\inl$ and $\inr$ indicate a left or a right choice. In the
case of $y_2$, the component is able to output a decision for each value received on $x$.

Similarly, we could implement $\role{Seller}$ and $\role{Shipper}$
with the components $C_{\role{Seller}}$ and $C_{\role{Shipper}}$:\label{example:zero}
\vspace{-6pt}
\begin{displaymath}
  \begin{array}{cc}
    \includegraphics[scale=.35]{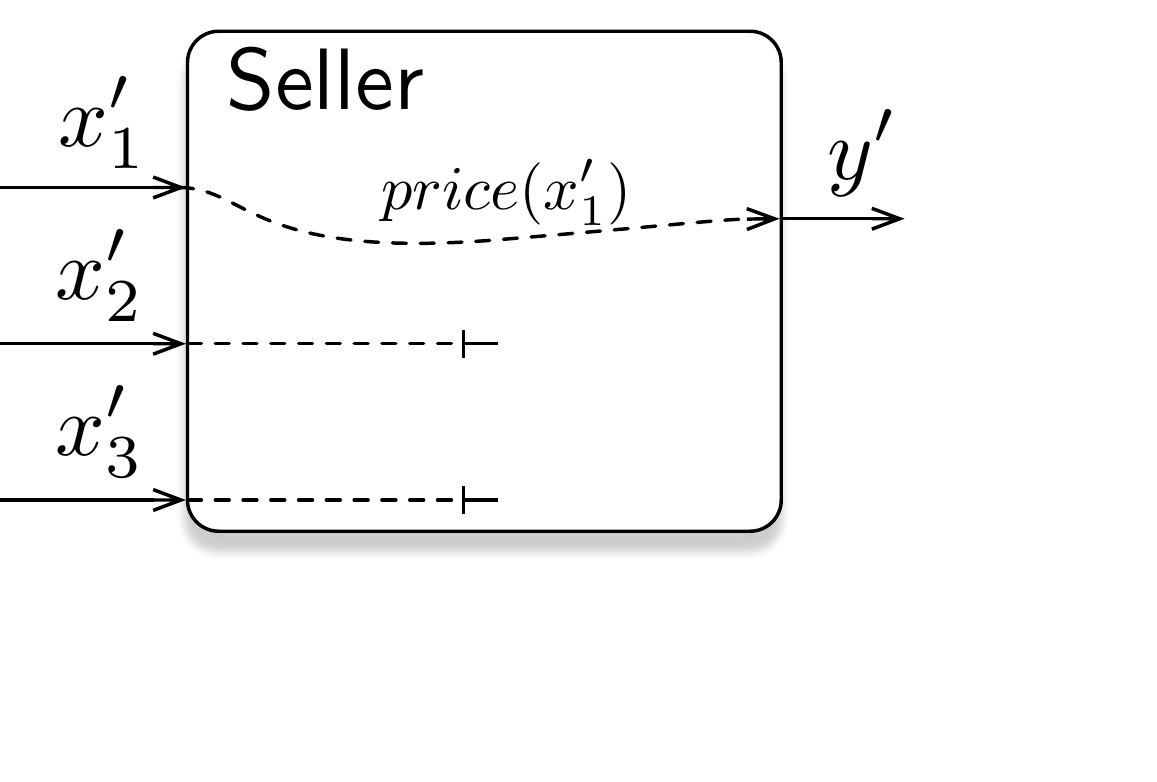}
    \qquad&\qquad
            \includegraphics[scale=.35]{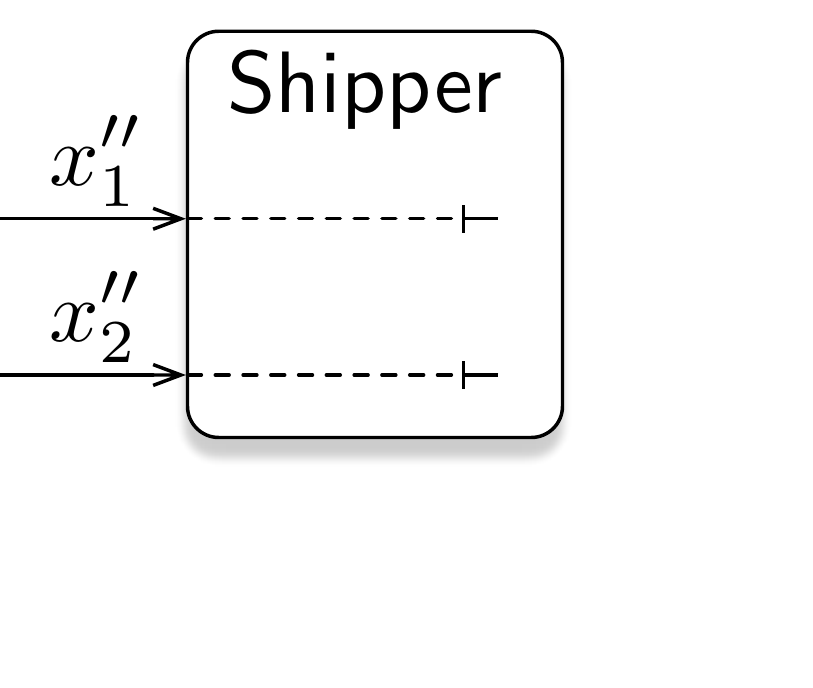}
  \end{array}
\end{displaymath}
\\[-36pt]
\noindent {}$\role{Seller}$ is a component that can always input on the
three ports $x_1'$, $x_2'$, and $x_3'$ and, for every message received
on $x_1'$ it can output on $y'$ the value $price(x_1')$, where $price$ is
a local function. $\role{Shipper}$, instead, only specifies the input receptiveness on $x_1''$
and $x_2''$.

The three components seen above are called {\em base
  components}. Using the BSS protocol, we can assembly them together,
obtaining the following {\em composite component}:
%

\label{example:one}
\includegraphics[scale=.35]{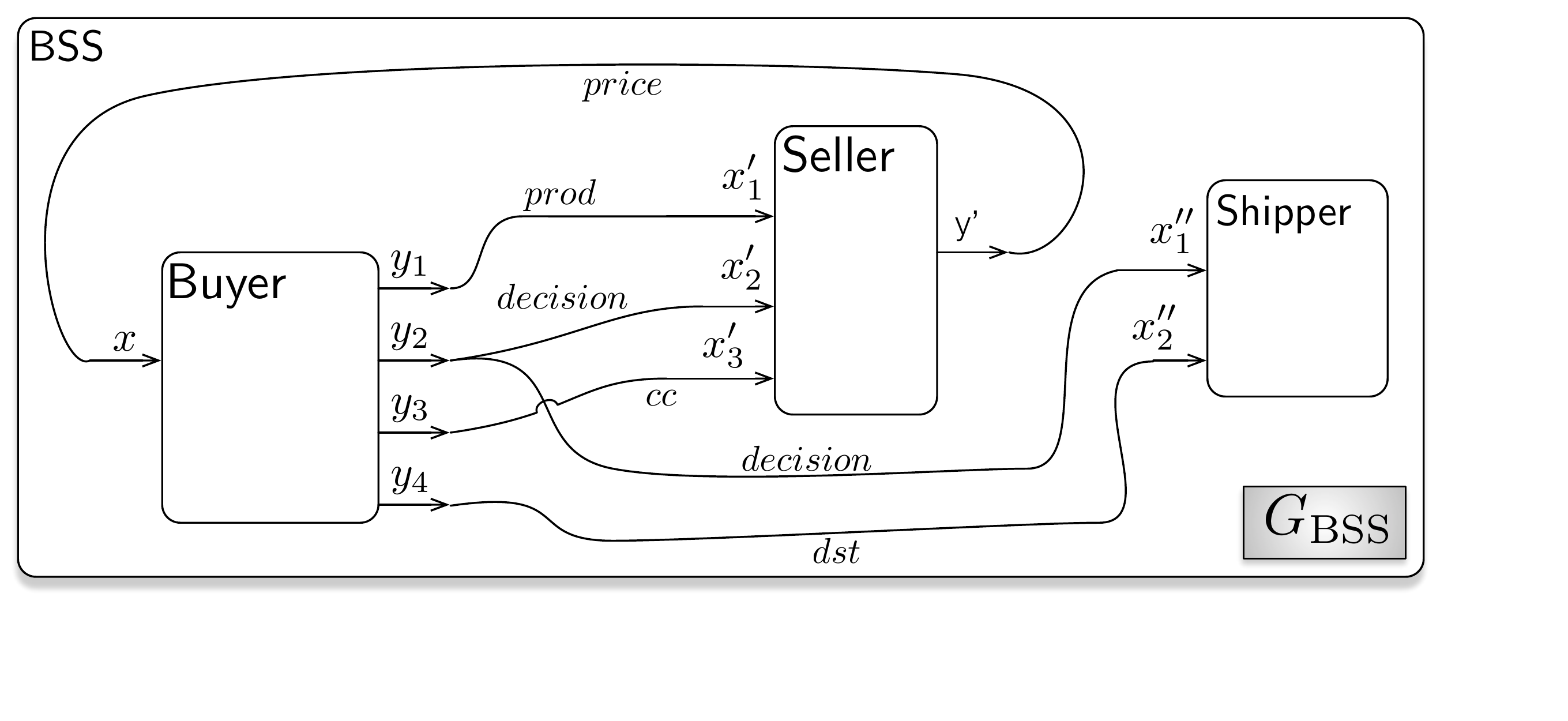}\vspace{-22pt}\\
The composite component above has three main ingredients: the protocol
$G_{\text{BSS}}$ that governs the internal communications, the
subcomponents $C_{\role{Buyer}}$, $C_{\role{Seller}}$, and
$C_{\role{Shipper}}$ that implement the roles in the protocol, and {\em
  connection binders} (full lines). The latter link outputs ports of a component to
input ports of other components. For example, the binder
$\dbinder{price}{Buyer}{x}{\roleport{Seller}{y'}}$ connects
$\role{Seller}$'s output port $y'$ to $\role{Buyer}$'s input port $x$:
here is how we ensure that the price given by $\role{Seller}$ reaches
$\role{Buyer}$, while decoupling the choreography specification from the actual
ports used in the implementation.


Now, suppose that we wish to replace the base component implementing
$\role{Seller}$ with a composite component that contains a 
 $\role{Sales}$ department and a $\role{Bank}$:
\\
{\vspace{-16pt}
  \includegraphics[scale=.37]{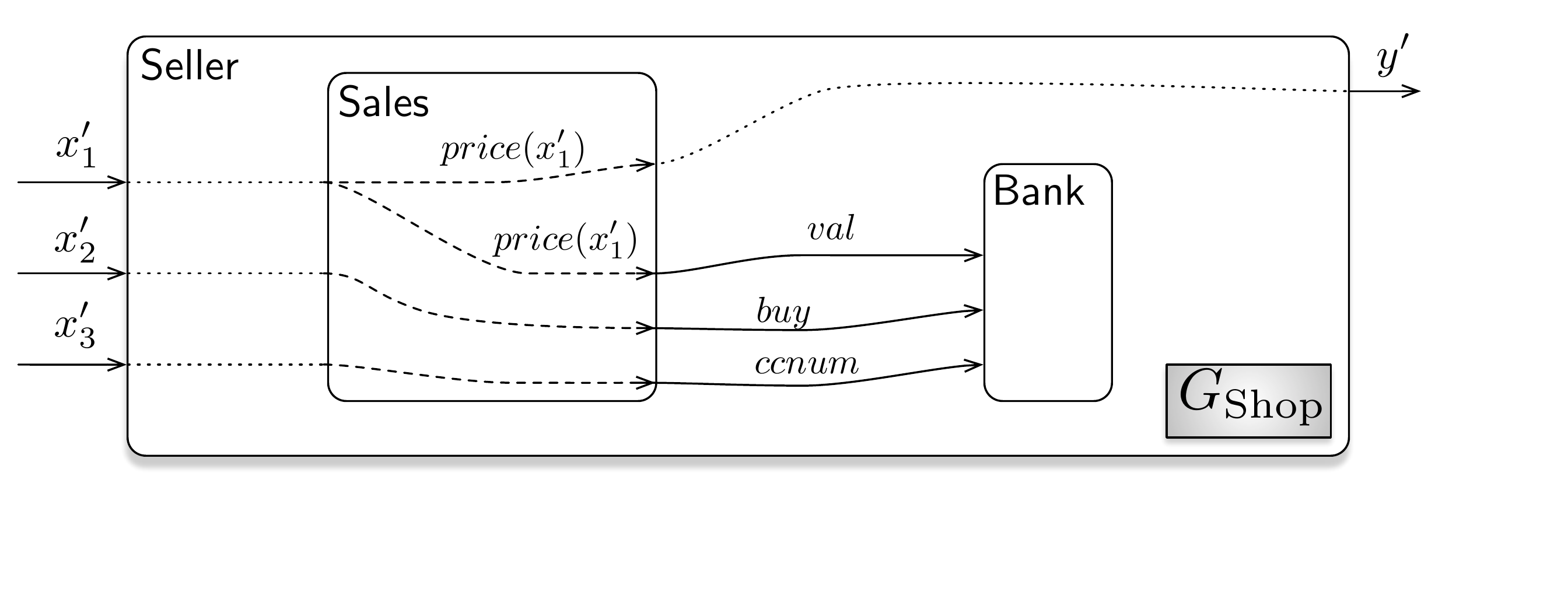}
}\vspace{-6pt}
 
\noindent The composite component above differs from the composite
component 
previously seen. First, it has a non-empty interface identical to that
of $C_{\role{Seller}}$. 
Second, it has a special role ($\role{Sales}$) that, apart from interacting
in the internal choreography, also
deals with messages originating/directed from/to the external interface. This is done through
{\em forwarders} (dotted lines) that link the interface of
$\role{Sales}$ with the external interface, yielding encapsulation. For example, one
of $\role{Sales}$'s input ports is linked to the interface variable
$x_1'$ meaning that any external request for a product is
immediately forwarded. Similarly, any output on $\role{Sales}$'s top
output port is forwarded to the output port $y'$. 
We remark that $\role{Sales}$ specifies an input port that is actually used
in two local binders, one directed to the output port $y'$ while the other 
is directed to message $\mathit{val}$.
The other
role $\role{Bank}$ is implemented by a component $C_{\role{Bank}}$
which we leave unspecified.
%
%
The component is governed by the protocol $G_{\role{Shop}}$, defined
as:
\vspace{-6pt}
\begin{displaymath}
\gchoice{Sales}{buy}{Bank}
{\gcom{Sales}{val}{Bank};\gcom{Sales}{ccnum}{Bank}}{ \tend}
\end{displaymath}
\vspace{-18pt}

\section{A Language for Modular Choreographies}
\label{sec:language}
%
We now move to the formal presentation of the syntax and semantics of Governed Components (GC).
%
%
The syntax of terms is displayed in \cref{fig:gc_syntax}.
\begin{figure}[t]
$$
\begin{array}{c}
  \begin{array}{l@{\qquad}lr@{\quad}ll@{\qquad}l}
    \text{Components}
    & C & ::= & \chorboxb {\tilde x}{\tilde y}{L} &  \text{base}\\
            && \pp & \chorbox {\tilde x}{\tilde y}{D}{\role{r}[F]}{R}{G}\ &
\text{composite}\\[1ex]
	\text{Local binders} & L & ::=  &
                \lbinderp{y} {\mathit{f}(\tilde{x})}\\
                && \pp & L,L
                & 
    \\[1ex]
    \text{Protocols} & G &::=  & \gcom p \lab{\tilde  q};G  & \text{communication}\\
            && \pp     & \gchoice p\lab{ \tilde q} {G}{G} & \text{choice}\\
            && \pp     & \mu\recvar X.G & \text{recursion}\\
            && \pp     & \recvar X & \text{recursion variable}\\
            && \pp     & \gend & \text{termination}
    \\[1ex]
    \text{Role assignments} & R & ::= & \roleas{p}{C} \\ && \pp & R,R &
    \\[1ex]
    \text{Connection binders} & D & ::= &
             \dbinder{\lab}{p}{x}
            {\roleport{q}{y}}  &
            \\
            && \pp & D,D
               \\[1ex]
    \text{Forwarders} & F & ::=  &
                \fbinder zw \\ && \pp  &  F,F
                & 
  \end{array}
\end{array}
$$
\caption{Governed Components, Syntax}
\label{fig:gc_syntax}
\end{figure}
Programs in GC are components, denoted $C$. Components communicate over the network by using
ports, ranged over by $x,y,u,z,w$. A component is either a base component or a composite
component. In both base and composite components, the definition of the component consists of two
parts: an interface $\interface{\til x}{\til y}$, which defines a series of input ports $\til
x$---used for receiving values from other components---and a series of output ports $\til y$---used
for sending values to other components; and an implementation (in curly brackets, $\{\cdots\}$),
which determines how ports are going to be used and how internal computation is performed.

A base component term $\chorboxb {\tilde x}{\tilde y}{L}$ denotes a component with interface
$\interface{\til x}{\til y}$ and an implementation that consists solely of some local binders $L$. A
local binder $\lbinderp{y}{\mathit{f}(\tilde{x})}$ is a reactive construct, in the sense that it
performs a computation in response to incoming messages on the ports $\til x$.
Intuitively, whenever all of its parameters can be instantiated by receiving
values on its input ports $\til x$, the value computed by function $f$ will be sent through the 
output port $y$.
We abstract from how functions are defined, and assume that they can be 
computed when supplied with the required parameters ($\til x$).

A composite component term
$\chorbox {\tilde x}{\tilde y}{D}{\role{r}[F]}{R}{G}$ denotes a
component 
obtained by composing other
components. We explain each subterm separately.
The protocol $G$ prescribes how the internal components are going to interact.
A protocol is a global description of communications among roles, ranged over by
$\role p,\role q,\role r$. In a communication term $\gcom p \lab{\tilde  q};G$, role $\role p$
communicates some value to the roles $\til{\role q}$ (we assume $\til{\role q}$ to be
nonempty) and then proceeds as protocol $G$. In term $\gchoice p\lab{ \tilde q} {G_1}{G_2}$, role
$\role p$ communicates to roles $\til{\role q}$ its choice of proceeding either according to
protocol $G_1$ or $G_2$. All communications are labelled with a label $\lab$, which uniquely
identifies the communication step in the protocol.
Terms $\mu\recvar X.G$ and $\recvar X$ are for, respectively, defining recursion variables and
their invocation. Term $\gend$ is the terminated protocol.

Given a protocol $G$, we need to define the implementation for each role in the protocol. This is
done in the subterm $R$. Specifically, $\roleas{p}{C}$ assigns the component $C$ as
the implementation for role $\role p$.
Once both protocol and implementation of each role are given, we 
define how each communication in the protocol is going to be
realised by the implementing components.  This is covered by the
connection binders $D$.  A connection binder
$\dbinder{\lab}{p}{x}{\roleport{q}{y}}$ connects the input port $x$ of
(the component implementing) role $\role p$ with the output port $y$
of (the component implementing) role $\role q$.  The label $\lab$ in
the connection binder declares which communication in the protocol of
the composite component will be realised by the message exchange
enacted by the binder.  This will be key to defining the dynamics of
composite components. The intuition is that when the protocol of a
composite component prescribes a communication with label $\lab$, the
composite component will enact the communication according to the
specification of the connection binder with the same label: by taking
a message available from the output port of the role on the right to
the input port of the role on the left of the binder.

The subterms that we described so far for composite components are used to define internal
behaviour. The last subterm---$\role{r}[F]$---instead, serves the purpose of defining
externally-observable behaviour. Namely, it allows the internal component $\role r$ to
interact with the external context of the composite component through the ports of the
latter. This is obtained by defining appropriate forwarders $F$, of the form $\fbinder zw$.
The syntax can be used to specify two kinds of forwarders that can either forward inputs or
outputs, depending on the port names of the composite component. An input forwarder binds an
input port of the internal component $\role r$, say $x$, to an external input port of the composite
component, say $x'$, written $\fbinder x{x'}$. This means that all incoming messages at $x'$ from
the outer context of the composite component will be made available to the sub-component
implementing role $\role r$ at its input port $x$.
Dually, an output forwarder binds an output port, say $y$, of the internal component to an output
port of the composite component, say $y'$, written $\fbinder {y'}{y}$. This makes all output
messages from the sub-component sent through $y$ available to the outer context of the composite
component through $y'$. Notice that the flow of messages in a forwarder is always from right to
left ($\fbinder{}{}$).
In the remainder, we abstract from ordering in $L$, $R$, $D$, and $F$ (implicit exchange).

\subsection{Semantics}

We give an operational semantics for GC in terms of a labelled transition system
(LTS). We write $C \lto\lambda C'$ when there is a transition from $C$ to $C'$
with label $\lambda$.
Labels can denote: (i) the output of a value $v$ through a port $y$, written ${y}!{v}$; (ii) the 
input of a value $v$ through a port $x$, written ${x}?{v}$; and (iii) an internal move, written 
$\tau$. The syntax of labels is thus: $\lambda ::= {y}!{v} \mid {x}?{v} \mid \tau$.
We present the rules that define $\lto{}$ in two parts,
respectively for base and composite components.
\subsubsection{Base components.}
\begin{figure}[t]
$$
    \begin{array}{c}
        \infer[\m{OutBase}]
        {
        \chorboxb {\tilde x}{\tilde y}{L}
        \lto{y!{v}{}}
        \chorboxb {\tilde x}{\tilde y}{L'}
        }
        {
        L
        \lto{y!{v}{}{}}
        L'
        & y \in \til y
        }
\qquad
        \infer[\m{InpBase}]
        {
        \chorboxb {\tilde x}{\tilde y}{L}
        \lto{x?v}
        \chorboxb {\tilde x}{\tilde y}{L'}
        }
        { L \lto{x?v} L' & x \in \til x }
\\[1.5ex]
           \infer[\m{LConst}]
            {
            \lbinder y{\cdot} {\mathit{f}()}
            \lto{{y}!{v}}
            \lbinder y{\cdot}{\mathit{f}()}}
        {f() \eval v}
\qquad
            \infer[\m{LOut}]
            {
            \lbinder y{\sigma,{\til\sigma}} {\mathit{f}(\til{x})}
            \lto{{y}!{v}}
            \lbinder y{\til \sigma}{\mathit{f}(\tilde{x})}
            }
        {
			\{\til x\} {=} \dom(\sigma) & f(\sigma(\til x)) \eval v
        }
\\[1.5ex]
\infer[\m{LInpNew}]
{
\lbinder y{\til\sigma}{f(\tilde{x})}
\lto{{x}?{v}}
\lbinder y{\til\sigma,\{x \mapsto v\}} {f(\tilde{x})}
}
{
x \in \bigcap_{\sigma_i \in \til \sigma}\dom(\sigma_i)
&
x \in \til x
}
\\[1.5ex]
	\infer[\m{LInpUpd}]
        {
        \lbinder y{\til\sigma_1,\sigma,\til\sigma_2} {f(\tilde{x})}
       \lto{{x}?{v}}
        \lbinder y{\til\sigma_1,\sigma[x \mapsto v],\til\sigma_2} {f(\tilde{x})}
        }
        {
		x \in \bigcap_{\sigma_i \in \til \sigma_1}\dom(\sigma_i)
		&
		x \not\in \dom(\sigma)
		&
		x \in \til x
        }
\\[1.5ex]
\infer[\m{LInpDisc}]
{
	\lbinder y{\til\sigma} {f(\tilde{x})}
       \lto{{x}?{v}}
    \lbinder y{\til\sigma} {f(\tilde{x})}
}
{
	x \not\in \til x
}
\\[1.5ex]
\infer[\m{LOutLift}]
       {
       L_1, L_2
       \lto{{y}!{v}}
       L_1', L_2
       }
       {
	L_1 \lto{{y}!{v}}L_1'}
\qquad
       \infer[\m{LInpList}]
       {
       L_1, L_2
       \lto{{x}?v}
       L_1', L_2'
       }
       {
	L_1 \lto{{x}?v}L_1'
\qquad
	L_2 \lto{{x}?v}L_2'}
    \end{array}
$$
  \caption{GC, semantics of base components.}
  \label{fig:gc_semantics_base}
\end{figure}
The rules defining the semantics of base components are displayed in \cref{fig:gc_semantics_base}.
Rule \rname{OutBase} states that if the local binders $L$ in a base component can produce a value 
$v$ for one of the output ports of the components $\til y$, then the component performs the 
corresponding output.
Rule \rname{InpBase} is similar, but models inputs instead. 

The other rules, whose names start with \rname{L}, define the
semantics of local binders, which are used by base components to
perform computation.  The key idea is that a local binder
$\lbinderp{y}{f(\til x)}$ can generate an output value for $y$ when a
value for each $x$ in $\til x$ has been received, as required by $f$.
Since we may need to receive on more than one input port before we can
compute $f$, we augment the syntax of local binders to runtime queues
of value stores: $L ::= \lbinder{y}{\til \sigma}{f(\til x)} \mid
L,L$. A store $\sigma$ is a partial mapping from ports to values,
which we use to store incoming messages until we have enough values to
compute $f$. In the remainder, we use $\lbinderp{y}{f(\til x)}$ in
programs as a shortcut for $\lbinder{y}{\cdot}{f(\til x)}$, where
$\cdot$ is the empty sequence.  We write $f(\til v) \eval v$ for
``computing $f$ by instantiating its parameters with the values
$\til v$ yields the value $v$'', abstracting from the concrete
definition of function evaluation ($\eval$).

Rule \rname{LConst} is the rule for binders that require no inputs (the function has no 
parameters), and thus can always output the value computed by the related function.
Rule \rname{LOut} is the output rule for local binders with functions that require parameters (we 
assume $\til x$ nonempty in this rule, to distinguish from rule \rname{LConst}). It 
states that if all parameters ($\til x$) of the function $f$ in the binder are instantiated by the 
first store available in the queue, then we can output the value computed by the function.

Rules \rname{LInpNew} and \rname{LInpUpd} receive a value for an input
port used by a local binder. Rule \rname{LInpNew} is used for creating
a new store (all current stores already have a value for the port on
which we receive). Note that $\til\sigma$ may be empty, so this could
be the first store that we are creating in the queue of the binder.
Specifically, we write $\{x \mapsto v\}$ for the store with a single
mapping, from $x$ to $v$.  Rule \rname{LInpUpd} deals with receiving a
value in a store that already exists. We check that all stores before
the one that we are going to update already have a value for the port
$x$ that we are receiving from
($x \in \bigcap_{\sigma_i \in \til \sigma_1} \dom(\sigma_i)$), and
that the store $\sigma$ that we are updating ($[x \mapsto v]$) does
not have a value for $x$ yet ($x \not\in \dom(\sigma)$)---this gives
us the first store in the queue without a value for $x$.  Rule
\rname{LInpDis} discards inputs on ports that are not needed by a
local binder.

Rule \rname{LOutLift} lifts outputs to groups of binders (if a binder can output a value, so can 
the group of binders). Rule \rname{LInpList} states that whenever we can receive a value for an 
input port, all binders are allowed to react to it.

\begin{example}
  Let us consider the base component $C_{\role{Seller}}$ from
  page~\pageref{example:zero}. Clearly, by applying rules
  $\m{LInpNew}$ and $\m{InpBase}$, a new message is added to the queue
  of $x_1'$. Then, we can use rules $\m{LOut}$ and $\m{OutBase}$ to
  output on $y'$.
\end{example}

\subsubsection{Composite Components.}
%
\begin{figure*}[t]
$$
    \begin{array}{c}
{
        \infer[\m{OutChor}]
        {\chorbox{\tilde x}{\tilde y}{\tilde D}{\role{r}[\tilde F]}{\roleas{p}{C},\tilde R}{G} 
\lto{\tau}
        \chorbox{\tilde x}{\tilde y}{\tilde D}{\role{r}[\tilde F]}{ \roleas{p}{C'},\tilde R}{G'}}
      {C \lto{{u}!{v}{}} C' \qquad 
      \tilde D = \dbinder{\lab}{q}{z} {\roleport{p}{u}}, \tilde D'
        \qquad
        G \lto{\labout{\role p}{\lab}{v}} G'
        }
        }
        \\\\
        {
        \infer[\m{InpChor}]
        {\chorbox{\tilde x}{\tilde y}{\tilde D}{\role{r}[\tilde F]}{ \roleas{q}{C}, \tilde R}{G} 
\lto{\tau}
        \chorbox{\tilde x}{\tilde y}{\tilde D}{\role{r}[\tilde F]}{\roleas{q}{C'}, \tilde R}{G'}}
        {
        C \lto{{z}?v} C' \qquad
        \tilde D = \dbinder{\lab}{q}{z} {\roleport{p}{u}}, \tilde D' \qquad
        G \lto{\labinp{\role q}{\lab}{v}} G'
}        }
\\\\
               \infer[\m{Internal}]
       {\chorbox{\tilde x}{\tilde y}{\tilde D}{\role{r}[\tilde F]}{\roleas{s}{C}, \tilde R}{ G}
       \lto{\tau}
       \chorbox{\tilde x}{\tilde y}{\tilde D}{\role{r}[\tilde F]}{ \roleas{s}{C'}, \tilde R}{G}}
       {C \lto{\tau} C'}
\\\\
       {
              \infer[\m{OutComp}]
       {
       \chorbox {\tilde x}{\tilde y}{\tilde D}{\role{r}[\tilde F]}{\roleas{r}{C}, \tilde R}{G}
       \lto{{y}!{v}{}}
       \chorbox {\tilde x}{\tilde y}{\tilde D}{\role{r}[\tilde F]}{\roleas{r}{C'}, \tilde R}{G}
       }
       {
       C\lto{{z}!{v}{}{}} C'
       \qquad
       \tilde F = \fbinder {y}z , \tilde F'
       \qquad
       y \in \til y
       }}
        \\\\
{
        \infer[\m{InpComp}]
        {\chorbox{\tilde x}{\tilde y}{\tilde D}{\role{r}[\tilde F]}{\roleas{r}{C},\tilde R}{G} 
\lto{{x}?v}
        \chorbox{\tilde x}{\tilde y}{\tilde D}{\role{r}[\tilde F]}{ \roleas{r}{C'},\tilde R}{G}}
        { C \lto{{z}?v} C'  \qquad \tilde F = \fbinder z{x}, \tilde F' \qquad x \in \til x
	}
        }
    \end{array}
$$
  \caption{GC, semantics of composite components.}
  \label{fig:gc_semantics_composite}
\end{figure*}
We now move to the semantics of composite components. The rules are displayed in 
\cref{fig:gc_semantics_composite}.

The key rules are \rname{OutChor} and \rname{InpChor}, which allow internal 
components (those inside of the composite) to interact.
Rule \rname{OutChor} allows an internal component to send a message to another internal component.
In the first premise, we require that the sender component can output some value, $v$, over some 
port, $u$ ($C \lto{{u}!{v}{}} C'$).
In the second premise, we check that port $u$ at the role of the sender component ($\role p$) is 
connected (through a connection binder) to a port of some other (role of a) component
($\tilde D = \dbinder{\lab}{q}{z} {\roleport{p}{u}}, \tilde D'$).
In the third premise, we check that the protocol of the composite allows for the message to be sent.
This is formalised by the protocol transition $G \lto{\labout{\role p}{\lab}{v}} G'$, which reads 
``the protocol prescribes a message send from role $\role p$ for interaction $\lab$ carrying 
value $v$''. We give the formal rules for protocol transitions below.
Rule \rname{InpChor} is similar to rule \rname{OutChor}, but models inputs instead of outputs. The 
premises are equivalent, only now they require labels for input actions instead of output actions, 
and the receiving component must be assigned to the receiver role of the connection binder instead 
of the sender role.

Rule \rname{Internal} is standard and allows for internal actions in
sub-components.

Rules \rname{OutComp} and \rname{InpComp} capture
externally-observable behaviour. \rname{OutComp} allows the
internal component that can interact with the environment (that with
the role declared in the last subterm, $\role r$ in the rule) to fire
an external output on one of the ports of the composite. The renaming
specified by the forwarder in $\til F$ is applied to convert the port
name of the internal component to one of the output ports in the
interface of the composite, since the environment expects outputs from
the composite on one of such ports.  Similarly, \rname{InpComp}
applies the same reasoning to externally-observable inputs.

Our semantics abstracts from the concrete rules that are used to
derive protocol transitions---we need protocol transitions in rules
\rname{OutChor} and \rname{InpChor}.  We only need that such
transitions have labels that are either of the form
$\labout{\role p}{\lab}{v}$ (read ``role $\role p$ sends value $v$ on
interaction $\lab$'') or of the form $\labinp{\role p}{\lab}{v}$ (read
``role $\role p$ receives value $v$ on interaction $\lab$''). We use
$\alpha$ to range over these labels in the remainder.  However, if we
want to reason about the behaviour and meta-theoretical properties of
components (as we are going to do in the remainder), we need to fix a
semantics for protocols. We adopt an unsurprising one, by adapting the
asynchronous semantics of choreographies
from~\cite{CM13,HYC16,CM17:ice} to our setting.  This method requires
augmenting the syntax of protocols with runtime terms (i.e., they are
not intended to be used by programmers using our model, only by our
semantics) to store intermediate communication states, as follows:

\vspace{4pt}
{\centerline{
$
    G ::= \ldots
     \pp  \gcom {}{\lab,v}{\tilde  q};G  
     \pp  \gchoice{}{\lab,v}{\tilde q}{G_1}{G_2}
$\vspace{4pt}}
 
The terms above denote intermediate states where a message has been
sent by the sender but still not received by the receiver (the first term is for a value 
communication, the second for a choice).
Given the augmented syntax for protocols, their semantics is given by the rules displayed in 
\cref{fig:g_semantics}.

\begin{figure}[t]
$$
  \begin{array}{c}
    \infer[\textsf{GSVal}]
    {
    \gcom p \lab{\tilde  q};\,G
    \ \lto{\labout{\role p}{\lab}{v}}\ 
    \gcom {} {\lab,v}{\tilde  q};\,G
    }
    {
    }
    \ \quad
           \infer[\textsf{GSChoice}]
           {
           \gchoice p \lab{\tilde  q}{G_1}{G_2}
           \ \lto{\labout{\role p}{\lab}{v}}\ 
           \gchoice {} {\lab,v}{\tilde  q}{G_1}{G_2}
           }
           {
           v\in\{\inl,\inr\}
           }
    \\\\
    \infer[\textsf{GRVal}]
    {
    \gcom {} {\lab,v}{\tilde  q,q};\,G
    \;\;
    \lto{\labinp{\role q}{\lab}{v}}\;\;
    \gcom {} {\lab,v}{\tilde  q};\,G
    }
    {
    \tilde{\role q} \mbox{ nonempty}
    }
    \ \quad
    \infer[\textsf{GRVal2}]
    {
    \gcom {} {\lab,v}{q};\,G
    \;\;\lto{\labinp{\role q}{\lab}{v}}\;\;
    G
    }
    {
    }
   \\\\
           \infer[\textsf{GRChoice}]
           {
           \gchoice {} {\lab,v}{\tilde  q,q}{G_1}{G_2}
           \;\;\lto{\labinp{\role q}{\lab}{v}}\;\;
           \gchoice {} {\lab,v}{\tilde  q}{G_1}{G_2}
           }
           {
           \tilde{\role q} \mbox{ nonempty}
           }
    \\\\
           \infer[\textsf{GRChoice2}]
           {
           \gchoice {} {\lab,v}{ q}{G_{\inl}}{G_{\inr}}
           \;\;\lto{\labinp{\role q}{\lab}{v}}\;\;
           G_v
           }
           {
           v \in\{\inl,\inr\}
           }
           \quad
            \infer[\textsf{GRec}]
           {
           \mu \recvar X. G \lto{\alpha} G'
           }
           {
	   G \subst{\mu \recvar X. G}{\recvar X} \lto{\alpha} G'
           }
    \\\\
    \infer[\textsf{GConc1}]
    {
    \gcom {p} \lab{\til{\role q}};\,G
    \;\;\lto{\alpha}\;\;
    \gcom {p} \lab{\til{\role q}};\,G'
    }
    {
    G \ \lto{\alpha}\ G'
    \quad
    \m{role}(\alpha) \not\in \role p,\til{\role q} 
    }
    \quad
           \infer[\textsf{GConc2}]
           {
           \gchoice {p} {\lab}{\til{\role q}}{G_1}{G_2}
           \;\;\lto{\alpha}\;\;
           \gchoice {p} {\lab}{\til{\role q}}{G_1'}{G_2'}
           }
           {
           G_1\lto\alpha G_1'
           \quad
           G_2 \lto\alpha G_2'
           \quad
           \m{role}(\alpha) \not\in \role p,\til{\role q} 
           }
    \\\\
    \infer[\textsf{GConc3}]
    {
    \gcom {} {\lab,v}{\tilde  q};\,G
    \;\;\lto{\alpha}\;\;
    \gcom {} {\lab,v}{\tilde  q};\,G'
    }
    {
    G\lto{\alpha} G'
    \quad
    \m{role}(\alpha) \not\in \til{\role q}
    }
    \quad
           \infer[\textsf{GConc4}]
           {
           \gchoice {} {\lab,\inl}{\tilde  q}{G_1}{G_2}
           \;\;\lto{\alpha}\;\;
           \gchoice {} {\lab,\inl}{\tilde  q}{G_1'}{G_2}
           }
           {
           G_1\lto\alpha G_1'
           \quad
           \m{role}(\alpha) \not\in \til{\role q}
           }
    \\\\
           \infer[\textsf{GConc5}]
           {
           \gchoice {} {\lab,\inr}{\tilde  q}{G_1}{G_2}
           \;\;\lto{\alpha}\;\;
           \gchoice {} {\lab,\inr}{\tilde  q}{G_1}{G_2'}
           }
           {
           G_2 \lto\alpha G_2'
           \quad
           \m{role}(\alpha) \not\in \til{\role q}
           }
  \end{array}
$$
\caption{GC, semantics of protocols.}
\label{fig:g_semantics}
\end{figure}

Rule \rname{GSVal} models outputs: an interaction labelled $\lab$ from a role $\role p$ to many 
roles $\til{\role q}$ can transition with label $\labout{\role p}{\lab}{v}$ to the runtime term
$\gcom {} {\lab,v}{\tilde  q}$. The runtime term registers that, on communication $\lab$, there is 
now an in-transit value $v$ that still has to be received by the roles $\til{\role q}$.
Rule \rname{GSChoice} is similar, but for choices. In particular, we require the transmitted 
value to be one of the two constants $\inl$ and $\inr$, since this value should inform the 
receivers of which continuation branch has been chosen by the sender ($G_1$ or $G_2$, 
respectively).

Rule \rname{GRVal} models input: if there is an in-transit message, this can be consumed by one 
of the intended receivers. Similarly for choices, in rule \rname{GRChoice}.
When there remains only one intended receiver, the protocol proceeds with the intended 
continuation, as specified by rules \rname{GRVal2} and \rname{GRChoice2}. Observe that, for 
choices, the value is used to determine which branch should be used to continue ($\inl$ goes
left, $\inr$ goes right).
Rule \rname{GRec} captures recursion in the standard way.

The rules prefixed with \rname{GConc} model concurrent and asynchronous execution of protocols, 
recalling full $\beta$-reduction in $\lambda$-calculus. In these rules, we use the function 
$\m{role}$ to extract the subject of a label. Formally: $\m{role}(\labout{\role p}{\lab}{v}) = 
\m{role}(\labinp{\role p}{\lab}{v}) = \role p$. Then, 
the \rname{GConc} rules cover all possible cases where the continuation of a 
protocol can execute an action that does not involve a role in the prefix of the protocol. Whenever 
this is the case, the continuation is allowed to make the transition without affecting the prefix. 
Intuitively, consider that 
two interactions that involve separate roles are 
non-interfering.

 \begin{example}
   Let us consider the BSS protocol implementation previously 
   described in pag.~\pageref{example:one}. Since
   $C_{\role{Buyer}}\lto{{y_1}!{\text{``The Winds of
         Winter''}}{}}C'_{\role{Buyer}}$ (for some $C'_{\role{Buyer}}$)
   and
   $G \lto{\labout{\role p}{prod}{\text{``The Winds of Winter''}}} G'$,
   we have that the whole component can reduce by rule $\m{OutChor}$
   because of the connection binder
   $\dbinder{prod}{Seller}{x_1'} {\roleport{Buyer}{y_1}}$. The new
   protocol $G'$ is obtained by replacing $G$'s first interaction with
   $\gcom {} {prod, \text{``The Winds of Winter''}}{Seller}$.
    \end{example}

For the purpose of the following sections and in order to simplify 
presentation we only consider well-formed components where
certain simple structural constraints are valid, namely: (1)
output ports and input ports are disjoint sets; (2) a port can only
be defined (lhs) once in local binders; (3) composite components
specify choreographies where sender is never included in receivers;
(4) forwarders given in composite components are defined for distinct
port identifiers; (5) composite components specify connection
binders where receiver ports are used at most once and
sender ports are always used together with the same message label.

%

\subsection{Examples}
\label{sec:examples}
We now describe a series of examples that inform on the semantics of
our model.

\smallskip

\noindent {\bf Example A.}
Let us consider a protocol where a role $\role p$ communicates a
decision to another role $\role q$ which then replies with different
messages according to such decision:
\[G=\gcom{p}{l_1}{q}\ (\gcom{q}{l_2}{p},\ \gcom{q}{l_3}{p})\]
We then consider the following component (dummy, for the sake of illustration): \\
\includegraphics[scale=.35]{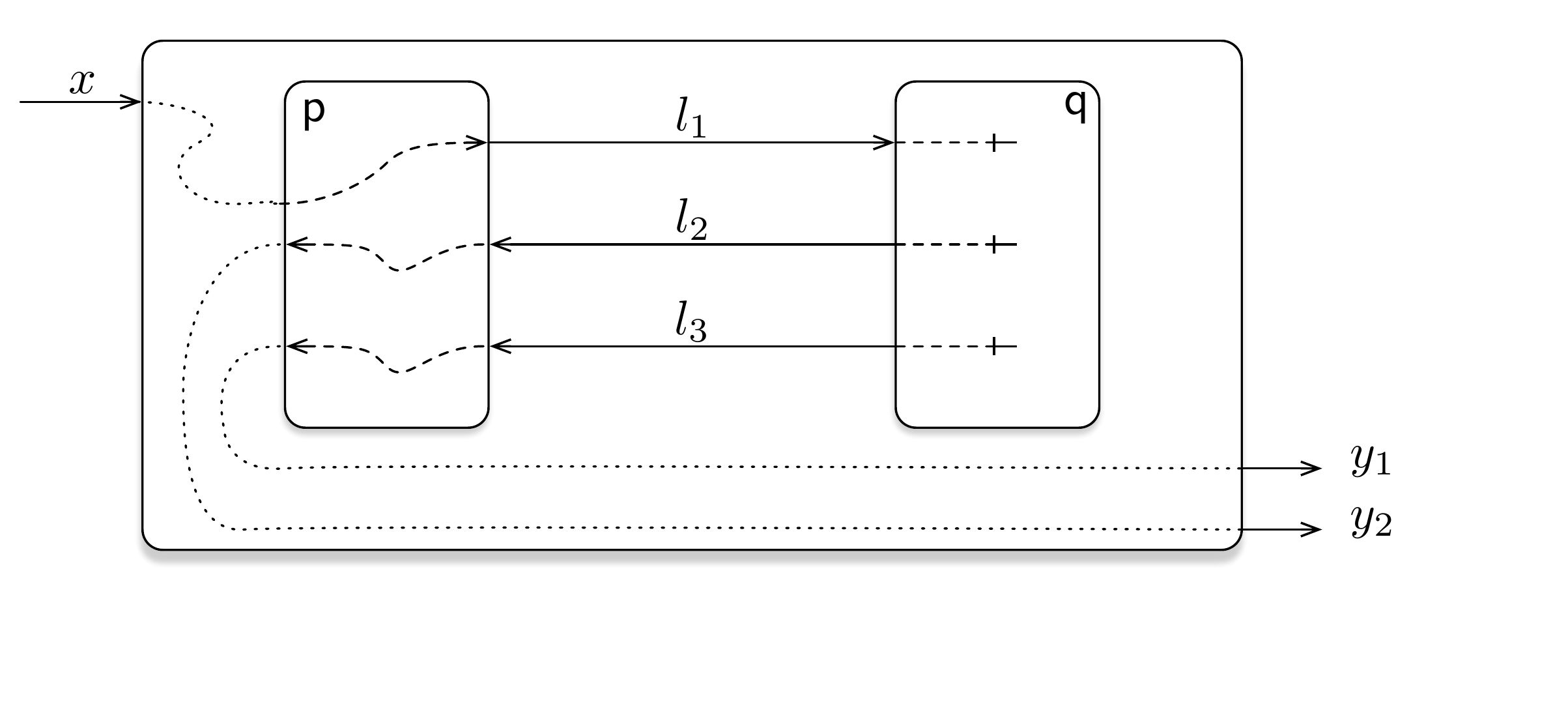}\\[-8mm]
Notice that the internal decision ($l_1$) actually originates from the external
interface (port $x$). 
In this case, we can see how the inner protocol $G$ refines the
internal behaviour, since an output will be observable from only one of the external ports 
$y_1$ and $y_2$, given the decision on $x$ and the fact that either message 
$l_2$ or $l_3$ is passed. This pattern is not possible to capture using 
a base component.

\smallskip

\noindent {\bf Example B.}
We now briefly change the protocol from the previous example,
where $\role p$, based on the decision communicated to $\role q$,
sends different messages to $\role q$: 
\[G=\gcom{p}{l_1}{q}\ (\gcom{p}{l_2}{q},\ \gcom{p}{l_3}{q})\] We may design 
a component where $\role p$ sends constants to $\role q$ which are 
forwarded to $y_1$,
$y_2$, and $y_3$:\\
\includegraphics[scale=.35]{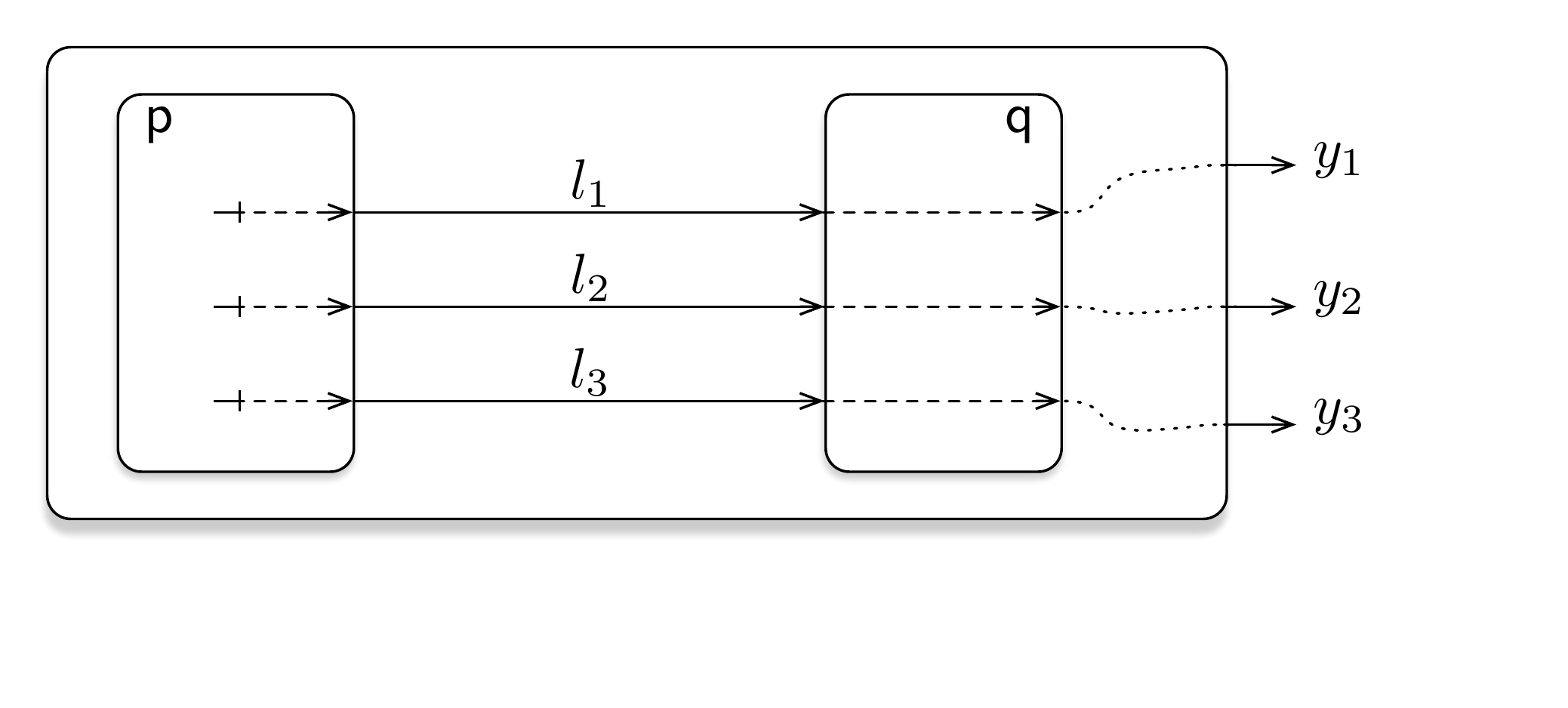}\\[-8mm]
Note that, depending on value emitted on $y_1$ there may
be a value available to be sent on $y_2$ or $y_3$, exclusively, similarly to Example A. The difference
is that here the decision is internal to the component and exposed via $y_1$.

\smallskip

\noindent {\bf Example C.}
Let us consider the two protocols
$G_0=\gcom{p}{l_1}{q};\ \gcom{q}{l_2}{p}$ and
$G_1=\mu\recvar X.\gcom{p}{l_1}{q};\ \gcom{q}{l_2}{p};\recx X$. Roles
$\role p$ and $\role q$ send a message to each other. In $G_1$ the
behaviour is repeated continuously, while in $G_0$ it happens only once. 
And let us consider the
component:\\
\includegraphics[scale=.35]{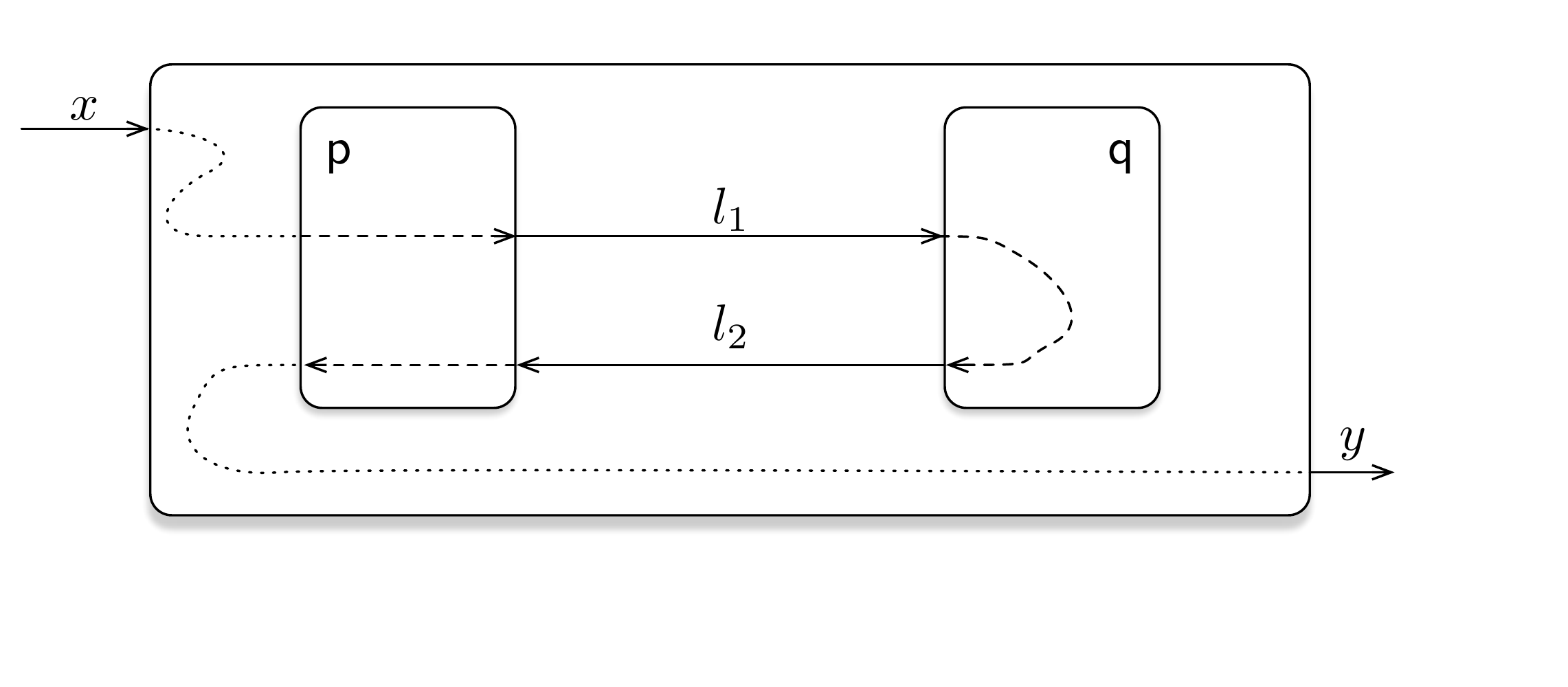}\\[-8mm]
If the component is governed by $G_0$, then we have a ``one shot''
component (one $x$ received, one $y$ emitted). However, if we use
$G_1$, the component becomes ``reactive'', i.e., every time a value is
received on $x$, a value on $y$ is emitted. So the component governed
by $G_1$ is as reactive as a base component, while the ``one shot''
usage is not representable via a base component.

\smallskip

\noindent {\bf Example D.}
%
Let us consider the following component together with the protocols from
example $C$:\\
\includegraphics[scale=.35]{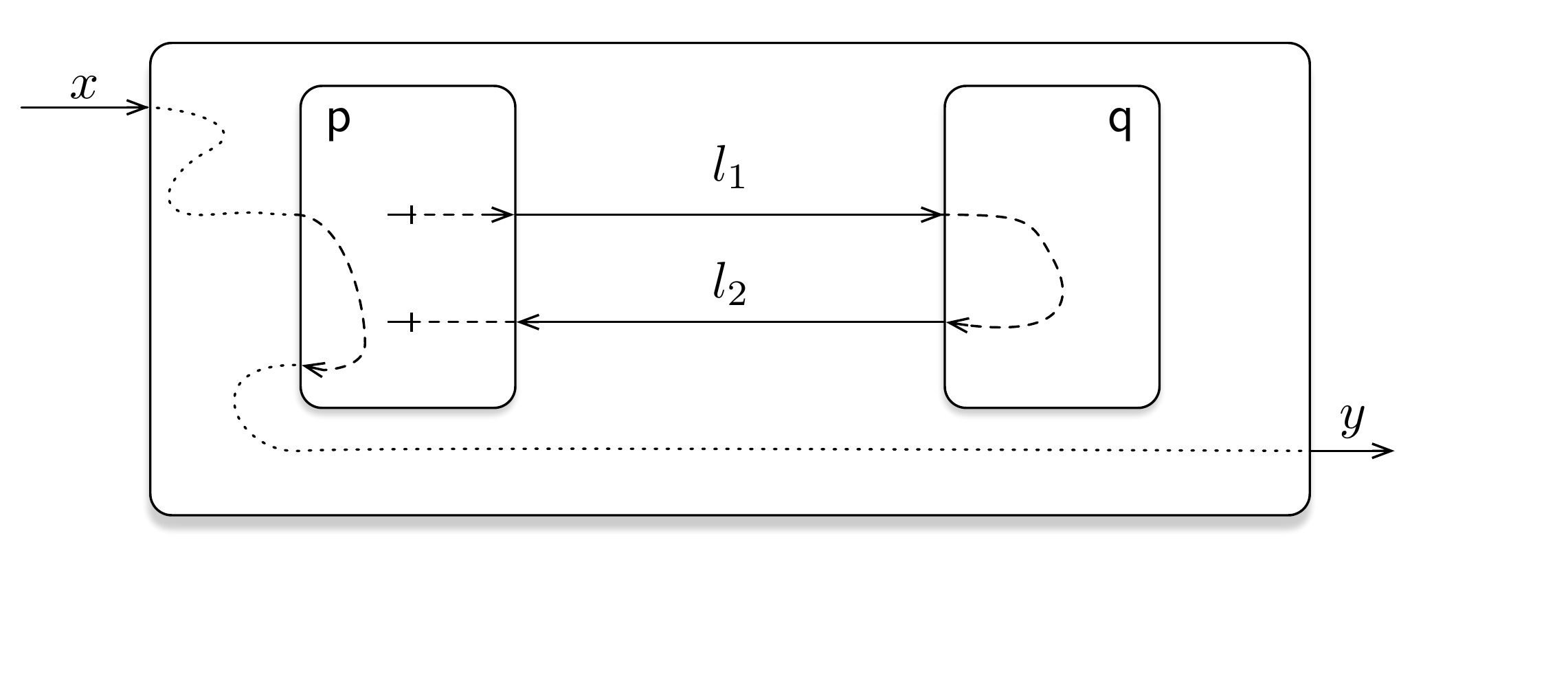}\\[-8mm]
In this example, we can see that the internal protocol does not interfere
with external reactive behaviour (no data passes from the outside to
the inner protocol). In such case reactions to the external environment
(for each $x$ received a $y$ may be emitted)
are completely independent from the internal communications (possibly several 
if the component is governed by $G_1$).

\smallskip

\noindent {\bf Example E.} In the protocol
$G=\gcom{p}{l_1}{q}; \mu\recvar X.\gcom{q}{l_2}{p};\recx X$, $\role p$
sends an initial message to $\role q$ who will then repeatedly
send a message to $\role q$. If used to govern the component:\\
\includegraphics[scale=.35]{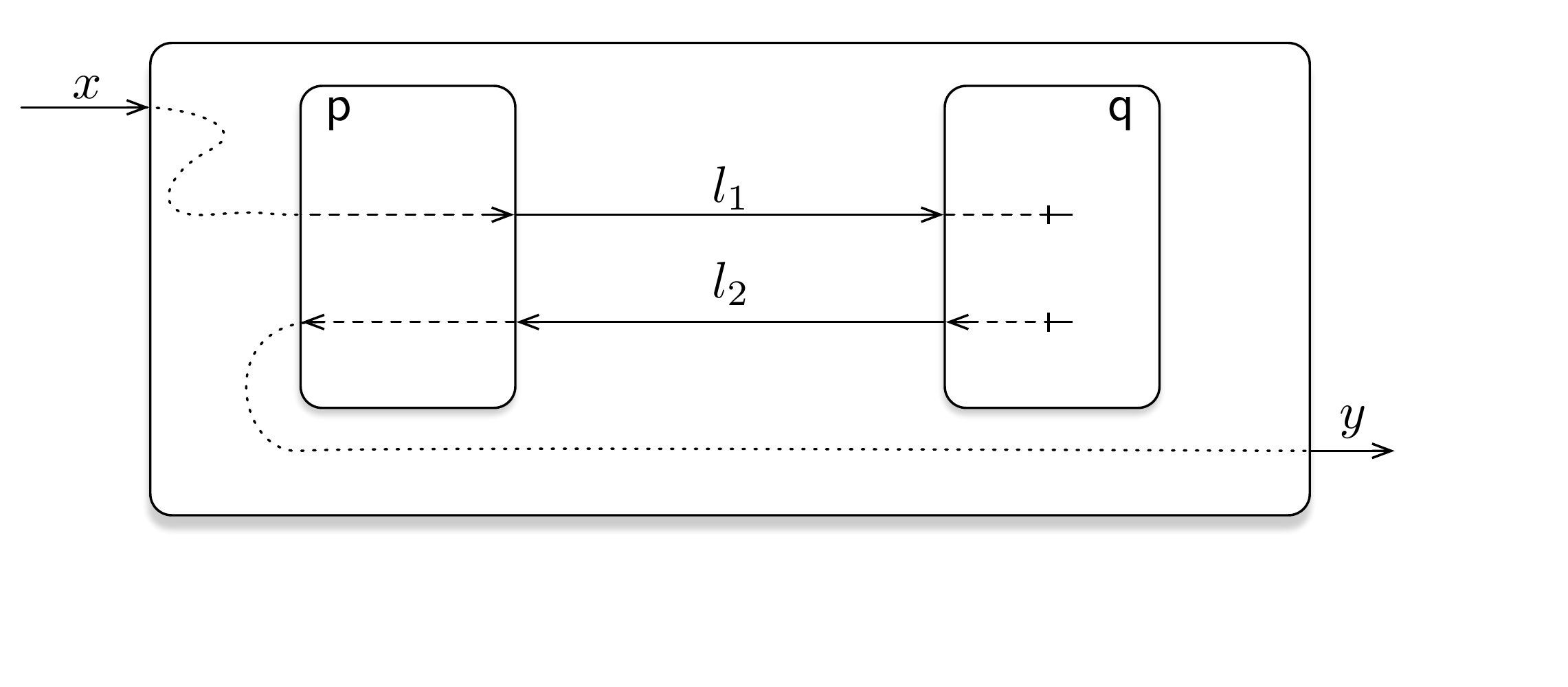}\\[-8mm]
we have the effect that receiving one $x$ kickstarts the internal interaction and
then an unbounded number of $y$'s can be emitted.

\smallskip

\noindent {\bf Example F.} We now consider a protocol that specifies
a choice similar to the one given in Example A, but wrapped in a recursion
so the protocol is repeating (in both branches):
\[G=\mu\recvar X.\gcom{p}{l_1}{q}\ (\gcom{q}{l_2}{p};\recx X,\
  \gcom{q}{l_3}{p}; 
  \recx X)\] We then consider that the protocol governs the following component:\\
  \includegraphics[scale=.35]{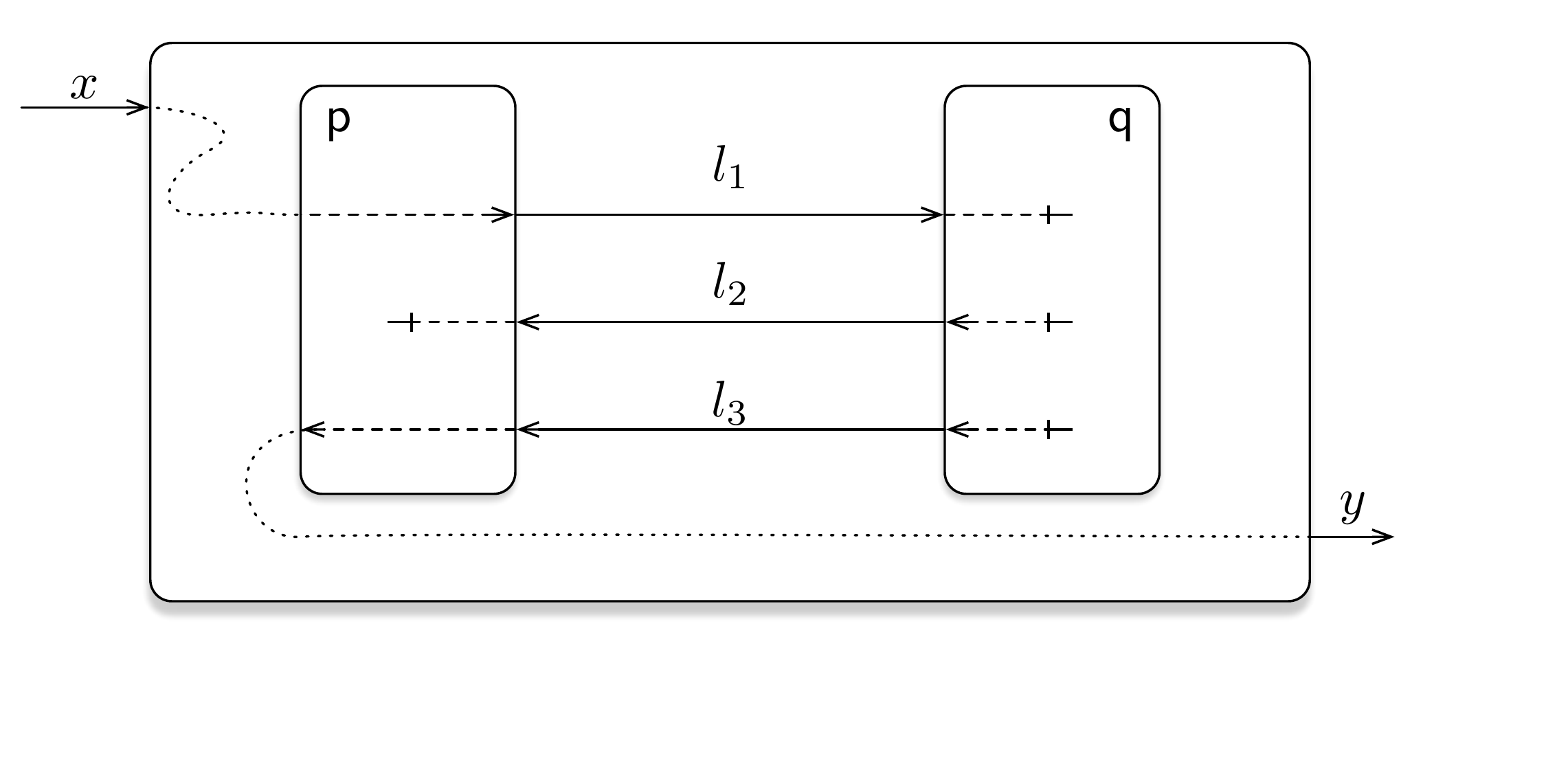}\\[-8mm]
  In this case, for each decision received on $x$ either there is no (external) output
  or a value is available to be sent on $y$.

%

\section{Target Process Model}
\label{sec:compilation}
In this section, we present an operational realisation of the GC
model, by showing how specifications can be compiled down to a more
foundational process model. The crucial point is to show that the centralised
control provided by protocols in GC can be realised in a distributed manner.
We consider a language that includes
communication primitives that are faithful with respect to the
original communication model, leaving out specialised constructs to
represent the structure of components. Our choice of target model
allows us, on the one hand, to tighten the gap towards a concrete
implementation, and, on the other hand, to show an operational
correspondence between source specifications and respective
compilations.

We adopt a variant of the $\pi$-calculus~\cite{SW01} where
communication is mediated by both queues and forwarders, and
(disciplined) non-determinism is supported via a specialised branching
primitive.  The syntax of processes ($P,Q,\ldots$) is given in
Fig.~\ref{fig:targetsyn}, where we (re)use port identifiers and
recursion variables.  In addition, we consider variables, ranged over
by $a, b, c, \ldots$, and objects ($o$) to represent both variables
and values.
\begin{figure}[t]
$$
  \begin{array}{rll}
    P &::=& 
		\zero \ \pp \ 
	 P \parop P \ \pp \ (\nub z) P \ \pp \ \mu \recx{X}. P \ \pp \ \recx{X}
	\ \pp \  z\msgout {o}. P \ \pp \ z\msgin {a}. P 
	  \\&  \pp &
	  z\chin (P,P) \ \pp \
	  \lbox{L} \ \pp \ (\fbinderl z z) P \ \pp \ (\fbinderr z z) P
  \end{array}
$$
\caption{Process Model Syntax}
\label{fig:targetsyn}
\end{figure}

The static fragment of the language is standard $\pi$-calculus. The
inactive process is represented by $\zero$, the parallel composition
$P \parop Q$ specifies that processes $P$ and $Q$ are simultaneously
active, restriction $(\nub z)P$ scopes port $z$ to process $P$, and
recursion is captured by the combination of the $\mu \recx{X}.P$ and
$\recx{X}$ constructs.
Communication primitives are specified using three language constructs: 
(1) The output $z\msgout{o}.P$ represents a process ready to send object 
$o$ (either a value or a variable that will be instantiated) on (port) ${z}$,
after which activating the continuation process $P$; 
(2) The input $z\msgin{a}.P$ captures a process ready to receive 
on $z$ after which activating the continuation process $P$, replacing 
the variable $a$ by the received value; 
(3) The branching $z\chin (P,Q)$ represents a process that receives
a selection value on ${z}$ and activates either $P$ or $Q$, depending
if the received value is $\inl$ or $\inr$, respectively.
The input ${z}\msgin{a}.P$ binds free occurrences of variable 
$a$ in $P$, $(\nub z)P$ binds occurrences of port $z$ in $P$ and 
$\mu \recx{X}.P$ binds occurrences of variable $\recx{X}$ in $P$. 
We use the usual abbreviation $(\nub \tilde z)P$.

The main difference with respect to the (synchronous) $\pi$-calculus
is that in our model communication is mediated by message queues,
hence output and input (and branching) primitives synchronise with
queues and not directly between them, similarly to the multiparty
session calculus~\cite{HYC16}. We reuse local binders for the purpose
of representing queues, since, on the one hand, they have a
straightforward operational interpretation and, on the other hand, it
allows us to compile GC base components in a direct way. We find
capturing the behaviour of local binders in different ways falls
outside the purpose of this paper. We represent by $\lbox{L}$ a
(delimited) set of local binders $L$. We also consider the constructs
$(\fbinderl z w) P$ and $(\fbinderr z w)P$ that forward the behaviours
that $P$ exhibits on $w$ as behaviours that $(\fbinderl z w) P$ and
$(\fbinderr z w)P$ exhibit on $z$, and exhibit the same behaviours for
other subjects. The difference is that $(\fbinderl z w) P$ forwards
exclusively outputs, while $(\fbinderr z w) P$ forwards solely inputs.

The semantics of the language is given by a labeled transition system 
defined by the rules shown in Fig.~\ref{fig:procsem} and the ones for
local binders shown in Fig.~\ref{fig:gc_semantics_base}, hence 
excluding rules $\m{OutBase}$ and $\m{InpBase}$ which refer to 
base components.
By 
$P\; \raisebox{-1.2pt}{\makebox{$\stackrel{\lambda}{\longrightarrow}$}} \; Q$
we denote that process $P$ evolves to process $Q$
exhibiting label $\lambda$. The set of labels, ranged over by 
(overloaded) $\lambda$, extends the set of labels of GC semantics 
by considering $\overline{\lambda}$ transitions, where $\lambda \neq \tau$,
hence 
$\lambda 
::= {y}!{v} \mid {x}?{v} \mid \tau \mid \overline{{y}!{v}} \mid \overline{{x}?{v}}$. 
This allows us to distinguish actions originating from communication
primitives ($\overline{\lambda}$) and queues ($\lambda$), so as to ensure 
primitives synchronise only with queues (and inversely). 

\begin{figure*}[t]
$$
\begin{array}{c@{\qquad}c@{\qquad}c}
\infer[\m{PAct}]
{ \lbox{L} \lto{\lambda} \lbox{L'}}
{L \lto \lambda L'}
&
\infer[\m{PRes}]
{(\nub z) P \lto{\lambda_{w}} (\nub z) Q}
{P \lto{\lambda_{w}} Q \qquad z \neq w}
\\[1.5ex]
\multicolumn{2}{c}{
\infer[\m{PPar}]
{P \parop R \lto{\lambda} Q \parop R}
{P \lto{\lambda} Q }
\qquad
\infer[\m{PCom}]
{P \parop Q \lto{\tau} P' \parop Q'}
{P \lto{\lambda} P' \qquad Q \lto{\overline{\lambda}} Q'}
\qquad
\infer[\m{PRec}]
{\mu \recx{X}. P \lto{\lambda} Q}
{P \subst{\mu \recx{X}. P}{\recx{X}}\lto{\lambda} Q }
}
\\[1.5ex]
\infer[\m{PFwd1}]
{(\fbinderr {z} {w}) P \lto{\lambda_{u}} (\fbinderr {z} {w})Q}
{P \lto{\lambda_u} Q \qquad z \neq u \neq w}
&
\infer[\m{PFwd2}]
{(\fbinderl {z} {w}) P \lto{\lambda_{u}} (\fbinderl {z} {w})Q}
{P \lto{\lambda_u} Q \qquad z \neq u \neq w}
\\[1.5ex]
\infer[\m{PFwdInp}]
{(\fbinderr {z} {w}) P \lto{z?v} (\fbinderr {z} {w})Q}
{P \lto{w?v} Q}
&
\infer[\m{PFwdOut}]
{(\fbinderl {z} {w}) P \lto{z!v} (\fbinderl {z} {w})Q}
{P \lto{w!v} Q}
\\[1.5ex]
\infer[\m{PFwdInp2}]
{(\fbinderr {z} {w}) P \lto{\overline{z?v}} (\fbinderr {z} {w})Q}
{P \lto{\overline{w?v}} Q}
&
\infer[\m{PFwdOut2}]
{(\fbinderl {z} {w}) P \lto{\overline{z!v}} (\fbinderl {z} {w})Q}
{P \lto{\overline{w!v}} Q}
\\[1.5ex]
\infer[\m{POut}]
{z\msgout{v}.P  \lto{\overline{z!v}} P }
{}
&
\infer[\m{PInp}]
{z\msgin{a}.P  \lto{\overline{z?v}} P \subst{v}{a}}
{}
\\[1.5ex]
\infer[\m{PChoL}]
{{z} \chin(P,Q)  \lto{\overline{z?\inl}} P}
{}
&
\infer[\m{PChoR}]
{{z} \chin(P,Q)  \lto{\overline{z?\inr}} Q}
{}
\end{array}
$$
\caption{Process Semantics.}
 \label{fig:procsem}
\end{figure*}

We briefly comment on the rules shown in Fig.~\ref{fig:procsem}, where
we use $\lambda_z$ to
identify a transition that specifies port $z$, be it
an input, an output or the respective $\overline{\lambda}$ variants. 
Rule $\m{PAct}$ specifies that a delimited set of local binders exhibits
the same behaviours as the local binders. Rule $\m{PRes}$ says
$(\nub z)P$ exhibits transitions of $P$ that specify subject $w$ different 
from the restricted $z$ (notice only values can be communicated).
Rule $\m{PPar}$ specifies the parallel composition exhibits transitions 
of one of the subprocesses. Rule $\m{PCom}$ captures the synchronisation 
between two processes, yielding a $\tau$ move of the parallel composition,
where one of the labels originates from a communication primitive 
$\overline{\lambda}$ and the other from a queue $\lambda$. 
We omit  the symmetric versions of rules $\m{PCom}$ and $\m{PPar}$. 
Rule $\m{PRec}$ says the recursive processes exhibits the transitions
of its unfolding.

The rules that start with $\m{PFwd}$ capture the 
behaviour of the forwarders. Rules $\m{PFwd1}$ and $\m{PFwd2}$
say that the forwarder constructs exhibit the same transitions as the scoped
process as long
as the port given in the transition is different from the ones involved in the 
forwarding. Rules $\m{PFwdInp}$ and $\m{PFwdInp2}$ capture input 
forwarding, while $\m{PFwdOut}$ and $\m{PFwdOut2}$ capture output
forwarding, where the transition on the (external) port $z$ is derived from 
a transition on the (internal) port $w$. We distinguish the cases of $\lambda$
($\m{PFwdInp}$ and $\m{PFwdOut}$)
and $\overline{\lambda}$ ($\m{PFwdInp2}$ and $\m{PFwdOut2}$)
to avoid introducing dedicated notation.
Rule $\m{POut}$ describes the behaviour of the output construct that exhibits
the corresponding output transition $\overline{z\msgout{v}}$, identifying that the
transition originates in a communication primitive. Rule $\m{PInp}$ shows the 
symmetric case. 
Rules $\m{PChoL}$ and $\m{PChoR}$
capture the behaviour of the branching construct, which is ready to receive
on $z$ a selection value, either $\inl$ or $\inr$, evolving to $P$ or to $Q$,
respectively.

Having presented the target language, we may now turn to the encoding
of GC, starting by the encoding of choreographies which are 
central in the operational model.
We denote by $\projb{G}_{\role r, D, \mapping}$ the encoding of 
choreography $G$ for role $\role r$ considering connection binders $D$
and mapping $\mapping$. The connection binders provide the association 
between message labels and communication ports. 
$\mapping$ is an injective mapping from role port pairs $(\role{r},z)$
and role label pairs $(\role{r}, \lab)$
into ports used in the target model, so as to omit 
role identifiers and map message labels, while ensuring a unique 
relation. 

\begin{figure}[t]
$$
\begin{array}{lclll}
\projb{\gcom p\lab{ \tilde q};G}_{\role p, D, \mapping} & \deff & 
z'\msgin{a}.
u\msgout{a}.
\projb{G}_{\role p, D, \mapping} 
\\ & & \hfill
(D = \dbinder {\lab}{ q}{ w}{\roleport {p}{z}}, D' 
\;\wedge\; 
 \mapping(\role{p},{z}) = z'
 \;\wedge\; 
 \mapping(\role{p},{\lab}) = u)
\\
\projb{\gcom p\lab{ \tilde q, q};G}_{{\role q},  D, \mapping}  & \deff & 
u\msgin{a}.
w'\msgout{a}.
\projb{G}_{\role q, D, \mapping} 
\\ && \hspace{1cm}\hfill
(
D = \dbinder {\lab}{q}{w}{\roleport {p}{z}}, D'
\;\wedge\;
\mapping(\role{q},{w}) = w'
\;\wedge\; 
\mapping(\role{q},{\lab}) = u)
\\
\projb{\gcom p\lab{ \tilde q};G}_{\role r, D, \mapping}  & \deff & 
\projb{G}_{\role r, D, \mapping}  
\hfill
(\role r \not \in \role p,\tilde{\role q}) 
\\
\projb{\gcom {}{\lab,v}{ \tilde q, q};G}_{{\role q}, D, \mapping}  & \deff & 
u\msgin{a}.
w'\msgout{a}.
\projb{G}_{{\role q}, D, \mapping} 
\\ & & \hfill
(
D = \dbinder {\lab}{q}{w}{\roleport {p}{z}}, D'
\;\wedge\;\mapping(\role{q},{w}) = w'
\;\wedge\; 
\mapping(\role{q},{\lab}) = u
)
\\
\projb{\gcom {}{\lab,v}{ \tilde q};G}_{\role r, D, \mapping}  & \deff & 
\projb{G}_{\role r,  D, \mapping} 
\hfill
(\role r \not \in \tilde{\role q})
\\
\projb{\gchoice p\lab{\tilde q}{G_1}{G_2}}_{\role p, D, \mapping}  & \deff & 
z'\chin(u\msgout{\inl}.\projb{G_1}_{\role p,  D, \mapping},
u\msgout{\inr}.\projb{G_2}_{\role p, D, \mapping})
\\ & &\hfill
(D = \dbinder {\lab}{ q}{w}{\roleport {p}{z}}, D'
\;\wedge \; \mapping(\role{p},{z}) = z'
\;\wedge \; \mapping(\role{p},{\lab}) = u)
\\
\projb{\gchoice p\lab{\tilde q, q}{G_1}{G_2}}_{{\role q}, D, \mapping} & \deff & 
u\chin ( w'\msgout{\inl}. \projb{G_1}_{{\role q}, D, \mapping},
w'\msgout{\inr}. \projb{G_2}_{{\role q}, D, \mapping}) 
&\\ & &\hfill
(
D = \dbinder {\lab}{q}{w}{\roleport {p}{z}}, \tilde D'
\;\wedge\; \mapping(\role{q},{w}) = w'
\;\wedge\; 
\mapping(\role{q},{\lab}) = u)
\\
\projb{\gchoice p\lab{\tilde q}{G_1}{G_2}}_{\role r, D, \mapping} & \deff & 
\projb{G_1}_{\role r, D, \mapping} 
\hfill
(\role r \not \in \role p,\tilde{\role q} 
\;\wedge\; \projb{G_1}_{\role r, D, \mapping} = \projb{G_2}_{\role r, D, \mapping}) 
\\
\projb{\gchoice {}{\lab,v}{\tilde q, q}{G_1}{G_2}}_{{\role q},  D, \mapping} & \deff & 
u\chin ( w'\msgout{\inl}. \projb{G_1}_{{\role q}, D, \mapping},
w'\msgout{\inr}. \projb{G_2}_{{\role q},  D, \mapping}) 
\\ & &\hfill
(
 D = \dbinder {\lab}{q}{w}{\roleport {p}{z}}, D'
 \;\wedge\; \mapping(\role{q},{w}) = w'
 \;\wedge\; 
\mapping(\role{q},{\lab}) = u)
\\
\projb{\gchoice {}{\lab,\inl}{\tilde q}{G_1}{G_2}}_{\role r, D, \mapping} & \deff & 
\projb{G_1}_{\role r, D, \mapping} 
\hfill (\role r \not \in \tilde{\role q})
\\
\projb{\gchoice {}{\lab,\inr}{\tilde q}{G_1}{G_2}}_{\role r, D, \mapping} & \deff & 
\projb{G_2}_{\role r, D, \mapping} 
\hfill (\role r \not \in \tilde{\role q})
\\
\projb{ \mu \recvar X.G}_{\role r, D, \mapping} & \deff & 
\mu \recvar {X}. \projb{G}_{\role r, D, \mapping} 
\hfill ( \role r \in \mathit{roles}(G))\\
\projb{ \mu \recvar X.G}_{\role r, D, \mapping} & \deff & 
\zero
\hfill (\role r \not \in \mathit{roles}(G))
\\
\projb{ \recvar X}_{\role r,  D, \mapping} & \deff & \recvar{X}  \\
\projb{ \gend}_{\role r,  D, \mapping} & \deff & \zero  
\end{array}
$$
\caption{Choreography Encoding.}
\label{fig:chorenc}
\end{figure}

In the GC model components evolve via actions prescribed by the
choreography. This is captured in the encoding via interaction between
the choreography encoding and the component encoding, presented
afterwards. Also, as usual, the encoding of choreographies is carried out for 
each participant, so choreographic monitoring is actually carried out in a 
distributed way. We therefore specify interaction points between the 
distributed monitors so as to ensure they evolve in a coordinated way.
Following these principles we may now comment 
on the definition of choreography encoding, defined inductively in the
structure of choreographies, for a given role, as shown 
in Fig.~\ref{fig:chorenc}.

The first case considers a communication projected for the
sender role, encoded as an input on $z'$ followed by an output on $u$
that forwards the received value,
and continuation defined as the respective projection of the continuation $G$.
We use the connection binders to identify the communication port $z$
associated to message $\lab$ (sender) and
 the mapping $\mapping$ to obtain port $z'$ (uniquely)
associated to the role $\role{p}$ and port $z$.
Intuitively, the process yielded by the encoding starts by receiving, on 
$z'$, the value from the component that implements the sender role.
The interaction between component and choreography encoding  
realises the choreographic monitoring. 
 The received value is then forwarded via an output on $u$, 
 obtained via the mapping $\mapping(\role{p},\lab)$, directed to the monitors
resulting from the projection for the (multiple) receivers, thus 
ensuring coordination between the distributed monitors.  
The one to many interaction is supported by local binders.

The second case considers a communication projected for the input role, 
following similar lines to the one for the sender. The difference is that
the resulting process first receives on $u$ the value originating from the 
monitor for the sender role, where $u$ is
given by the mapping $\mapping$
for role $\role{q}$ and label $\lab$. 
The received value is then
forwarded using the port $w'$ associated to the respective role-port pair
of the component that implements the receiver role. The (third) case for a 
role not involved in the communication, is defined as the
respective projection of the continuation of the choreography. The (fourth and
fifth) cases for when the value is already in transit (i.e., when the emitter has 
sent the value), for roles waiting to receive the message and for any other role, 
are the same. 

When the choreography specifies a choice the communication pattern is
identical, i.e., components interact with choreography monitors and the 
latter interact between them. The difference is that the possible values are 
known and the implied alternative behaviours are implemented, via 
the branching construct (of the form $z\chin(P, Q)$). 
For the case of the sender role, if the component, via $z'$, sends 
$\inl$ then the left branch
is taken and $\inl$ is forwarded to receiver monitors on $u$, and 
likewise for $\inr$. For the case of the receiver role, the monitor
starts by receiving the selection value from the dedicated queue, given
by the mapping $\mapping(\role{q},\lab)$. If the value 
is $\inl$ then the left branch is taken and $\inl$ is sent to the receiver
component, via $w'$, and
likewise for $\inr$. For roles not involved in the choice we adopt 
the (simplest) standard condition (cf.~\cite{HYC16})
that says the encodings of the two alternative branches must be
identical. The case for choices when the value is in transit for roles waiting to 
receive the selection value is the same as for the receiver role mentioned 
above. For other roles, which include the sender that has sent the value and 
roles that have already received the selection, the encoding directly considers 
the relevant branch since the value is registered in the choreography.

The last cases show the encoding for recursion and termination, which are
represented directly by the respective constructs of the target language. 
In the case of recursion, we ensure that the projection for a given role 
of the recursion body is actually meaningful (when the role participates 
in the choreography), otherwise it is directly discarded so as to avoid
undesired configurations. 
We remark
that the encoding is a partial function, given
that the specified conditions must hold.

As mentioned previously, we use local binders to support the coordination
between the monitor associated to the sender role and the monitors
associated to the receiver roles. Intuitively, there is a queue for 
each receiver and each queue is able to input the values provided by the 
monitor for the sender. The distribution to several receiver queues is ensured 
by local binder semantics. Given connection binders $D$ and
mapping $\mapping$ we use $\buildqueues{D}{\mapping}$ to denote
the set of local binders (queues) defined as:
%
$\buildqueues{\dbinder {\lab}{q}{w}{\roleport {p}{z}}, D'}{\mapping}
\deff
(\lbinder {\mapping(\role{q},\lab)}{\cdot} {\mathit{id}(\mapping(\role{p},\lab))}),
\buildqueues{D'}{\mapping}$}.

Hence, there is a local binder for each connection binder, defined using 
the identity function ($\mathit{id}$) for the purpose of value forwarding. 
Each local binder 
inputs messages originating from the sender monitor,
using the port given by $\mapping(\role{p},\lab)$, and 
outputs messages towards the receiver monitor,
using the port $\mapping(\role{q},\lab)$.
Since choreographies themselves may carry (in transit) values,
we introduce an operation that places the values in the respective
queues, denoted by $\queues{G}{L}{\mapping}$ for a given choreography 
$G$, local binders $L$ and mapping $\gamma$ (defined in expected lines,
see Appendix~\ref{app:defs}). We remark that such operation targets only
choreography runtime terms and is a homomorphism for other cases,
up to a side condition for choices that says that the result of
applying the operation to both branches is equal.

We now present the encoding of components which is shown in 
Fig.~\ref{fig:compenc}. In order to match GC component behaviour
and ensure a modular specification,
the encoding yields processes that have as set of free ports ($\mathit{fp}$)
only identifiers in the interface $\tilde x, \tilde y$ and can
only exhibit inputs on $\tilde x$ ports and exhibit outputs 
on $\tilde y$ ports. The former is ensured by a (total) renaming
given by a mapping ($\mapping$) combined with the appropriate name 
restrictions, while the latter is enforced via the forwarding constructs.
So, firstly, the encoding of a base component $\chorboxb {\tilde x }{\tilde y}{L}$ 
considers $\mapping$ using a (dummy) role $\role r$ to replace all
identifiers given in the source GC specification, and specifies a 
name (list) restriction  for all ports yielded by the mapping (the 
distinguished set $\tilde z$). 
The original local 
binders are also renamed using the mapping
(denoted $\mapping(\role r, L)$). Secondly, the encoding
specifies input forwarders for $\tilde x$ and output forwarders for $\tilde y$
targeting their mappings. 

\begin{figure}[t]
$$
\begin{array}{lcll}
\projb{\chorboxb {\tilde x }{\tilde y}{L}} & \deff & 
(\nub \tilde z) \\ & & 
(\fbinderr  { x_1} {\mapping(\role r, x_1)}) 
\ldots
(\fbinderr  { x_k} {\mapping(\role r, x_k)})\\ & & 
(\fbinderl {y_1} {\mapping(\role r, y_1)})
\ldots
(\fbinderl {y_n} {\mapping(\role r, y_n)})\\ & & 
\lbox{{\mapping(\role r},L) }
\\ & & \hfill
( 
\tilde x = x_1, \ldots, x_k 
\,\wedge\, \tilde y = y_1,\ldots,y_n
\, \wedge\, 
\tilde z \cap (\tilde x, \tilde y) = \emptyset 
\\ & & \hfill 
\, \wedge \, \tilde z = \mathit{co}(\mapping) 
\,\wedge \,
\mathit{dom}(\mapping) = \role r \times (\mathit{fp}(L)\cup \tilde x \cup \tilde y)
)
\\
\projb{\chorbox {\tilde x}{\tilde y}{ D}{\role r[F]}{R}{G}} & \deff & 
(\nub \tilde z)\\ & & 
(\fbinderr { x_1} {\mapping(\role r, w_1)})
\ldots
(\fbinderr { x_k} {\mapping(\role r, w_k)}) \\ & & 
(\fbinderl {y_1} {\mapping(\role r, z_1)})
\ldots
(\fbinderl {y_n} {\mapping(\role r, z_n)})
\\ & & 
( \projb {R}_\mapping
\parop 
  \projb G_{{ {\role p}_1}; D; \mapping} 
\parop \ldots \parop
 \projb G_{{ {\role p}_j}; D; \mapping}
\parop \lbox{\queues{G}{ (\buildqueues{D}{\mapping}) }{\mapping}} )
\\ & & \hfill
(
x_1, \ldots, x_k   \subseteq \tilde x \,\wedge\,
y_1, \ldots, y_n  \subseteq \tilde y \, \wedge\, 
\tilde z \cap (\tilde x, \tilde y) = \emptyset 
\\ & & \hfill
\,\wedge\, 
F =  \fbinder {w_1}{x_1} , \ldots, \fbinder {w_k}{x_k},
\fbinder {y_1}{z_1}, \ldots, \fbinder {y_n}{z_n} 
\\ & & \hfill 
\,\wedge\,
\tilde z = \mathit{co}(\mapping) 
\,\wedge\, \mathit{dom}(\mapping) = (\mathit{rplp}(D) \cup 
\role{r} \times \tilde w \cup \role{r} \times \tilde z)
\\ & & \hfill
\wedge
\, {\role p}_1,\ldots,{\role p}_j = \mathit{roles}(G) )
\\
\\
\projb{\roleas{r}{C}}_\mapping & \deff &
(\nub \tilde x,\tilde y) \\ & &
(\fbinderr {\mapping(\role r, x_1)}{x_1} )
\ldots
(\fbinderr {\mapping(\role r, x_k)}{x_k} )\\ & &
(\fbinderl {\mapping(\role r, y_1)}{y_1} )
\ldots
(\fbinderl {\mapping(\role r, y_n)}{y_n} )\\ & &
\projb{C}
\\ & &  \hfill
(C = \chorboxb {\tilde x }{\tilde y}{\ldots}
\, \wedge \,
\tilde x = x_1, \ldots, x_k 
\,\wedge\, \tilde y = y_1,\ldots,y_n)
\\
\\
\projb{R_1, R_2}_\mapping & \deff & 
\projb{R_1}_\mapping \parop \projb{R_2}_\mapping
\end{array}
$$
\caption{Component Encoding.}
\label{fig:compenc}
\end{figure}

The encoding of the composite component follows the same principles,
specifying a complete identifier replacement and restricting all ports not
part of the interface. Also forwarders are specified, following precisely
the forwarders given in the GC specification, up to the respective mapping.
The elements of the component are then encoded as the parallel composition
of the monitor queues with the encodings of role assignments and 
the choreography. The monitor queues are specified by the 
$\buildqueues{\_}{\_}$, considering the original connection binders and the
mapping, up to the placement of in transit values given by $\queues{\_}{\_}{\_}$.
The encoding of the choreography
is carried out for all roles specified in the choreography, considering
the original connection binders and the mapping. 
The encodings of role assignments consider the respective mapping, 
where the ports used in the mapping in 
the context of the composite component are forwarded to the ports 
specified in the interface of the (sub)component. The conditions specified
for the encoding ensure that the renaming is total, that the yielded
processes operate on a distinguished set of ports, and only the ports
in the interface are free ports of the process. We use $\mathit{rplp}(D)$
to denote the set of role ports and role labels used in $D$.

We may now present our results that show there is a precise operational
correspondence between GC specifications and their encodings, starting
by ensuring each (internal) evolution in the source specification is matched 
by precisely two (internal) evolutions in the target model. For the purpose
of our results we use a standard notion of structural congruence 
(denoted $\equiv$), extended with forwarder swapping (provided the
names involved in the forwarding are distinct).

\begin{proposition}[Encoding Soundness]
\label{pro:compsound}
If $C \lto{\tau} C'$ then $\projb{C} \lto{\tau}P\lto{\tau} \equiv \projb{C'}$.
\end{proposition}
\begin{proof}
By induction on the derivation of $C \lto{\tau} C'$.
\end{proof}

Proposition~\ref{pro:compsound} builds on the auxiliary
results, which show that component (non-internal) transitions are
matched precisely by their encodings, and that component evolutions 
are matched by two actions of their encodings (see Appendix~\ref{app:encres}).
The encoding requires an extra step with respect to the original 
specification for the purpose of the coordination of the distributed monitors
(while in GC semantics the monitoring is carried
out in a centralised way). 
We now state the completeness result for our encoding, where 
we use $\ltow{\tau}$ to represent zero or more $\tau$ transitions. 


\begin{proposition}[Encoding Completeness]
\label{cor:compcomp}
If $\projb{C} \ltow{\tau} P$ then there is $C'$ such that $C \ltow{\tau} C'$ and 
$P \ltow{\tau}\equiv \projb{C'}$.
\end{proposition}
\begin{proof}
By induction on the derivation of $\projb{C} \ltow{\tau} P$.
\end{proof}

Proposition~\ref{cor:compcomp} says that 
any configuration reachable by the encoding 
may always evolve to (the encoding of) a configuration reachable by the component.
%
We thus have that behaviours of
the encoding are matched by the source component and also that
the encoding does not diverge.

The operational results attest the correctness of our encoding, 
ensuring that component and encoding behaviour match.
However correctness of component themselves is not addressed, 
so for instance if the component has some undesired behaviour,
also will the encoding. In order to single out components that
enjoy desired properties we introduce a typing discipline, described 
next, which results can be carried down to the 
compilation thanks to the operational correspondence.

%


 \section{Types for Components}

In this section we present our type system that addresses communication 
safety and progress, and also a preliminary investigation on the notion of
substitutability supported by our model. We start by mentioning some 
notions auxiliary to the typing.
First of all we introduce the type language that captures component 
behaviour, equipped with an operational semantics for the purpose of showing 
the correspondence with GC semantics. We refer to our types as \emph{local 
types} since they can be obtained by \emph{projecting} the protocol 
considering a specific role (cf.~\cite{HYC16}). In order to analyse composite 
components we also have to \emph{merge} two local types, in particular the 
result of a protocol projection and an (external) interface type. Finally, to 
capture base component behaviour we define an \emph{abstract} 
version of local binders, equipped with an operational semantics, used to 
check the \emph{conformance} of a base component w.r.t. a local type.

The syntax of local types, given
in Fig.~\ref{fig:localtypes}, builds on communication actions carried out on ports.
We assume a set of base types, ranged over by $B$ (e.g., $\mathtt{INT}$).
Output type $\tout z B T$ describes a component that sends a value of type $B$ on port ${z}$ and that 
afterwards follows specification $T$; likewise for the symmetric input type $\tinp z B T$. 
Choice type $\tchoice{z}{T_1}{T_2}$ describes a component that sends a selection value, either $\inl$ 
or $\inr$, on port ${z}$ and that afterwards follows specification $T_1$ or $T_2$, respectively. 
The symmetric branch type $\tbranch{z}{T_1}{T_2}$ describes a component that receives a selection value, either $\inl$ 
or $\inr$, on port ${z}$ and that afterwards follows specification $T_1$ or $T_2$ respectively. 
The type of the communicated values in the case of choice and branch types is implicit, but since it will be useful later on,
we introduce $\chot$ as the type for $\inl$ and $\inr$.
Local types include standard recursion constructs and termination that is specified by $\tend$. 
By convention we identify alpha-equivalent recursive types.

\begin{figure}[t]
$$
  \begin{array}{rllllll}
\quad T & ::= & \tout z B T & \text{(output)} 
& \pp & \tinp z B T & \text{(input)} \\
& \pp & \tchoice{z}{T}{T} & \text{(choice)}
& \pp & \tbranch{z}{T}{T} & \text{(branch)}\\
& \pp & \mu \recvar X. T & \text{(recursion)} 
& \pp & \recvar X & \text{(variable)}\\
& \pp & \tend & \text{(termination)} \quad \\
\end{array}
$$
\caption{Local Types}
\label{fig:localtypes}
\end{figure}

The labelled transition system given in Fig.~\ref{fig:localsem} defines the semantics of local types, intuitively explained above.
Since we are interested in showing the correspondence of local types and component behaviour,
the set of labels considered is the same considered for GC (except for $\tau$). 
Rule $\m{TOut}$ says an output type exhibits an output on the respective port of any value of the respective type,
after which activating continuation $T$; likewise for $\m{TInp}$. Rules for choice and branch types specify the
continuation to be activated according to the sent/received value. The
rule for recursion is standard.

\begin{figure}[t]
$$
  \begin{array}{cc}
  \multicolumn{2}{l}{
\infer[\m{TOut}]
{\tout y B T \lto{y!v} T}
{
v: B
}
\qquad
\infer[\m{TInp}]
{\tinp x B T \lto{x?v} T}
{
v: B
}
\qquad
\infer[\m{TRec}]
{ \mu \recvar X. T \lto{\lambda} T'}
{T \subst{\recvar X}{ \mu \recvar X. T} \lto{\lambda} T'}
}
\\[4mm]
\qquad
\infer[\m{TChoL}]
{\tchoice y {T_1} {T_2} \lto{y!\inl} T_1}
{
}
&
\infer[\m{TChoR}]
{\tchoice y {T_1} {T_2} \lto{y!\inr} T_2}
{
}
\\[4mm]
\qquad
\infer[\m{TBraL}]
{\tbranch x {T_1} {T_2} \lto{x?\inl} T_1}
{
}
&
\infer[\m{TBraR}]
{\tbranch x {T_1} {T_2} \lto{x?\inr} T_2}
{
}
\end{array}
$$
\caption{Local Type Semantics}
\label{fig:localsem}
\end{figure}

In GC the interaction in a composite component is controlled by a protocol, hence inner components
must support the behaviour expected by the protocol for the corresponding role. We thus require an
operation that projects the expected behaviour (i.e., local type) of a component given the protocol and 
the respective role. Also, the operation requires the relations between the protocol message labels and their 
respective communication ports together with the types of the communicated values. We denote by 
$\projc{G}{\role p}{D}{\labenv}$ the projection of protocol $G$ for role $\role p$ considering connection binders $D$ 
and mapping $\labenv$. Connection binders provide the association between message labels and communication ports (like in the encoding),
while $\labenv$ maps message labels to the (base) types of the communicated
values, allowing to ensure that both sender 
and receivers agree on the type of the value.
We remark that projection is a partial function since some conditions are necessary for the projection to exist.


We briefly present the projection operation, defined inductively on the structure of the protocol as shown in 
Fig.~\ref{fig:chorproj}. Protocol 
$\gcom p \lab{\tilde q};G$ specifies that 
$\role p$ sends message $\lab$ to $\tilde {\role q}$, so
the projection for the emitter role $\role p$ is an output type $\tout z B {(\ldots)}$, 
considering a respective connection binder $ \dbinder \lab { q}{w}{\roleport{p}{z}}$ to identify (output) port $z$, and mapping 
$\Delta$ to identify the type ($B$) of the value. The continuation of the output type is the projection of the continuation $G$. The 
projection for any of the receiver roles follows similar lines, differing in the resulting (input) type and the (input) port considered 
in the connection binders. The projection for a role $\role r$, different from $\role p$ and $\tilde {\role q}$, considers directly the projection of the continuation 
since $\role r$ is not involved in the message exchange. 

\begin{figure}[t]
$$
\begin{array}{lclll}
\projc{(\gcom p \lab{\tilde q};G)}{\role p}{D}{\labenv} & \triangleq & 
\tout z B \projc{G}{\role p}{D}{\labenv}
\\ & & \hspace{2cm}\hfill
( D =  \dbinder \lab { q}{w}{\roleport{p}{z}},  D' \; \wedge \; 
\labenv(\lab) = B) 
\\
\projc{(\gcom p \lab{\tilde q, q};G)}{{\role q}}{ D}{\labenv} & \triangleq & 
\tinp {w} B \projc{G}{{\role q}}{ D}{\labenv}
\\ & & \hfill
(
 D =  \dbinder \lab { q}{w}{\roleport{p}{z}},  D' \; \wedge \; 
\labenv(\lab) = B) 
\\
\projc{(\gcom p \lab{\tilde q};G)}{{\role r}}{ D}{\labenv} & \triangleq 
& \projc{G}{{\role r}}{ D}{\labenv}
\hfill ({\role r} \not \in {\role p}, \tilde{{\role q}})\\
\projc{(\gcom {} {\lab,v}{\tilde q,q};G)}{{\role q}}{ D}{\labenv} & \triangleq 
& \tinp {w} B \projc{G}{{\role q}}{ D}{\labenv} 
\\ & & \hfill
(
D = \dbinder \lab {q}{w}{\roleport{p}{z}},  D' 
\;\wedge\; v: B )
\\
\projc{(\gcom {} {\lab,v}{\tilde q};G)}{{\role r}}{ D}{\labenv} & \triangleq & 
\projc{G}{{\role r}}{ D}{\labenv}
\hfill ({\role r} \not \in  \tilde{{\role q}})
\\
\projc{(\gchoice p\lab{\tilde q} {G_1}{G_2})}{{\role p}}{ D}{\labenv} 
& \triangleq & \tchoice{z}{\projc{G_1}{{\role p}}{ D}{\labenv}}
{\projc{G_2}{{\role p}}{ D}{\labenv}}
\\ & & \hfill
( D = \dbinder \lab {q}{w}{\roleport{p}{z}},  D' ) 
\\
\projc{(\gchoice p\lab{\tilde q, q} {G_1}{G_2})}{{\role q}}{ D}{\labenv}
& \triangleq & \tbranch{w}{\projc{G_1}{{\role q}}{ D}{\labenv}}
{\projc{G_2}{{\role q}}{ D}{\labenv}} 
\\ & & \hfill
(
D = \dbinder \lab {q}{w}{\roleport{p}{z}},  D')
\\
\projc{(\gchoice p\lab{\tilde q} {G_1}{G_2})}{{\role r}}{ D}{\labenv} 
& \triangleq & \projc{G_1}{{\role r}}{ D}{\labenv}
\\ & & \hfill
 ({\role r} \not \in {\role p}, \tilde{{\role q}}\;\wedge\;
\projc{G_1}{{\role r}}{ D}{\labenv} = \projc{G_2}{{\role r}}{ D}{\labenv})
\\
\projc{(\gchoice {}{\lab,v}{\tilde q, q} {G_1}{G_2})}{{\role q}}{ D}{\labenv} 
& \; \triangleq \; & 
\tbranch{w}{\projc{G_1}{{\role q}}{ D}{\labenv}}{\projc{G_2}{{\role q}}{ D}{\labenv}} 
\\ & & \hfill 
(
D = \dbinder \lab {q}{w}{\roleport{p}{z}},  D'
 \;\wedge\; v: \chot)
\\
\projc{(\gchoice {}{\lab,\inl}{\tilde q} {G_1}{G_2})}{{\role r}}{ D}{\labenv} 
& \triangleq & 
\projc{G_1}{{\role r}}{ D}{\labenv} 
\hfill ({\role r} \not \in  \tilde{{\role q}})\\
\projc{(\gchoice {}{\lab,\inr}{\tilde q} {G_1}{G_2})}{{\role r}}{ D}{\labenv}  
& \triangleq & \projc{G_2}{{\role r}}{ D}{\labenv} 
\hfill ({\role r} \not \in  \tilde{{\role q}})\\
\projc{(\mu \recvar X.G)}{{\role r}}{ D}{\labenv} 
 & \triangleq & \mu \recvar X.(\projc{G}{{\role r}}{ D}{\labenv} )
\hfill ({\role r} \in \mathit{roles}(G))\\
 \projc{(\mu \recvar X.G)}{{\role r}}{ D}{\labenv} 
 & \triangleq & \tend
\hfill ({\role r} \not \in \mathit{roles}(G))\\
\projc{\recvar X}{{\role r}}{ D}{\labenv}  & \triangleq & \recvar X
\\
\projc{\gend}{{\role r}}{ D}{\labenv}  & \triangleq & \tend
\end{array}
$$
\caption{Protocol Projection.}
\label{fig:chorproj}
\end{figure}

The projection of protocol  $\gcom {} {\lab,v}{\tilde q};G$, where value $v$ is in 
transit, has two cases: for any of the roles waiting to receive the message, which follows similar lines to the input but where the mapping 
$\Delta$ is not used since (the type of) the value is known; for any other role, which is defined as the projection of the continuation.
In the latter case, notice that any other role may refer for instance to the emitting role, who has already sent the message. The
respective type at this stage does not specify the output since the protocol no longer expects the output from the respective 
component (the value is already in transit), and analogously for roles that have already received the message.

The projection of protocol $\gchoice p\lab{\tilde q} {G_1}{G_2}$ also considers the respective connection binders to identify
the ports used by the emitter and by the receivers, yielding choice or branch types respectively. In both cases the continuations 
are given by the projections of the respective continuations of the protocol $G_1$ and $G_2$. For a role $\role r$ not involved in the 
selection, the projection is defined 
only if the projection of the continuations $G_1$ and $G_2$ for role $\role r$ is the same. When the selection value is in transit in protocol  
$\gchoice {}{\lab,v}{\tilde q} {G_1}{G_2}$, the projection for roles waiting to receive the selection follows the same lines as before, where
we check that the value in transit is of the appropriate type ($\chot$). For other roles, since the selection is known, the projection
considers the implied continuation, which is crucial for the sake of typing preservation given that, for instance,
the sender has already committed on the selected branch.
%
The remaining cases for the projection are direct, where we ensure
(like in the encoding) that the projection for the recursive protocol is meaningful.

Protocol projection is a crucial operation when typing the composite component,
allowing to specify the behaviour of all inner components. However, in the case 
of the
distinguished component that interacts also with the external context, the expected behaviour is given by a 
combination of both the protocol projection for the respective role and the 
(external) behaviour of the composite component itself. We introduce an auxiliary notion 
that allows us to build such combination of two local types, a ternary relation which we 
refer to as merge, defined by the rules given in Fig.~\ref{fig:merge}. By $ T_1 \join T_2 = T_3$ 
we denote that $T_1$ can be merged to $T_2$ yielding type $T_3$. Intuitively, the
rules describe (all possible ways on) how to shuffle the communication actions of the
two source types so as to yield the combined one.

\begin{figure*}[t]
$$
\begin{array}{cc}
\infer[\m{MOutL}]
{\tout{z}{B}{T_1} \join T_2 = \tout{z}{B}{T_3}}
{T_1 \join T_2 = T_3}
& 
\infer[\m{MOutR}]
{T_1 \join \tout{z}{B}{T_2} = \tout{z}{B}{T_3}}
{T_1 \join T_2 = T_3}
\\\\
\infer[\m{MInpL}]
{\tinp{z}{B}{T_1} \join T_2 = \tinp{z}{B}{T_3}}
{T_1 \join T_2 = T_3}
& 
\infer[\m{MInpR}]
{T_1 \join \tinp{z}{B}{T_2} = \tinp{z}{B}{T_3}}
{T_1 \join T_2 = T_3}
\\\\
\infer[\m{MChoiceL}]
{\tchoice{z}{T_1}{T_2} \join T_3 = \tchoice{z}{T_1'}{T_2'}}
{T_1 \join T_3 = T_1' \qquad T_2 \join T_3 = T_2'}
\quad & 
\infer[\m{MChoiceR}]
{T_1 \join \tchoice{z}{T_2}{T_3} = \tchoice{z}{T_2'}{T_3'}}
{T_1 \join T_2 = T_2' \qquad T_1 \join T_3 = T_3'}
\\\\
\infer[\m{MBranchL}]
{\tbranch{z}{T_1}{T_2} \join T_3 = \tbranch{z}{T_1'}{T_2'}}
{T_1 \join T_3 = T_1' \qquad T_2 \join T_3 = T_2'}
\quad & 
\infer[\m{MBranchR}]
{T_1 \join \tbranch{z}{T_2}{T_3} = \tbranch{z}{T_2'}{T_3'}}
{T_1 \join T_2 = T_2' \qquad T_1 \join T_3 = T_3'}
\\\\
\infer[\m{MEndL}]
{\tend \join T = T}
{\mathit{fv}(T) = \emptyset}
&
\infer[\m{MEndR}]
{T \join \tend = T}
{\mathit{fv}(T) = \emptyset}
\\\\
\infer[\m{MRec}]
{\mu \recvar X. T_1 \join \mu \recvar X.T_2 = \mu \recvar X. T_3}
{T_1 \join T_2 = T_3}
&
\infer[\m{MVar}]
{\recvar X \join \recvar X = \recvar X}
{}
\end{array}
$$
\caption{Type Merge}
\label{fig:merge}
\end{figure*}

We briefly comment on the rules of Fig.~\ref{fig:merge}. $\m{MOutL}$ and $\m{MOutR}$
yield an output type, originating from left and right source types respectively, where the
yielded continuation is obtained by merging the continuation of the original output with 
the other source type. Rules $\m{MInpL}$ and $\m{MInpR}$ follow similar lines for the input type.
We notice that, in rules $\m{MChoiceL}$, $\m{MChoiceR}$, $\m{MBranchL}$ and $\m{MBranchR}$,
to merge choice and branch types we consider the other source type ($T_3$ in left rules, $T_1$ in
right rules) is merged to both continuations so as to yield the resulting continuations. 
Rules $\m{MEndL}$ and $\m{MEndR}$ say that merging $\tend$ with some other type $T$ yields type $T$,
provided $T$ has no free recursion variables (we use $\mathit{fv}$ to refer to this set) so as to exclude
merging non recursive types with recursive types. Rules $\m{MRec}$ and 
$\m{MVar}$ specify the merging of recursive behaviour, which must be matched on both source sides 
both at the beginning of the recursion and at the end, yielding corresponding types. We assume the 
merge operation is defined only when the sets of ports used by the source types are disjoint.

Protocol projection and merge are used in the typing 
of composite components, identifying the expected behaviour for each subcomponent,
while accounting for the external interface. For base components instead we must check the 
compatibility of the local type prescribed for the component (the external interface) with the 
local binders. To this end we define an abstract version of the local binders so as to capture 
their operational behaviour at static compilation time. On the one hand, we abstract away 
from values known only at runtime but, on the other hand, account for selections that can 
be determined statically. In particular, we account for choices that are defined exclusively on 
branching by keeping track of the dependency, allowing us to address choice propagation
between different components. 

Abstract local binders are denoted with $\absL$, defined following the lines
of local binders 
$\absL ::= \lbinder{y} {\tilde{\sigma} } {\mathit{f}(\tilde{x})} \; \pp\;  \absL,\absL$,
where we extend the stores $\tilde \sigma$ to map variables not only to values
but also to base types.
The semantics for abstract local binders is given in Fig.~\ref{fig:abssem},
following closely the one given in Fig.~\ref{fig:gc_semantics_base}. 
We use $\absv$
to denote either a base type $B$, choice type $\chot$, or $\inl$ - $\inr$ values,
used in the transitions to represent both a carried type or a value.
The main difference w.r.t. local binder semantics is given by rules 
$\m{AOutVal}$ and $\m{AOutType}$ where in the former we are
able to statically evaluate the function, in which case necessarily the store 
maps the relevant variables exclusively to values, and the transition exhibits the 
respective value (hence exactly the same as before). The latter
instead considers the case when it is not possible to evaluate the function,
hence the store contains mapping of relevant values to base types,
in which case the transition considers the return type of the function. Also
$\m{AConst}$ considers the return type of the function, leaving
the handling of constants open to runtime functions. The 
remaining rules abstract away from values or base types ($\absv$).

\begin{figure*}[t]
$$
\begin{array}{c}
           \infer[\m{AConst}]
            {
            \lbinder y{\cdot} {\mathit{f}()}
            \lto{{y}!{B}}
            \lbinder y{\cdot}{\mathit{f}()}}
        {f : \_ \rightarrow B}
\qquad
            \infer[\m{AOutVal}]
            {
            \lbinder y{\sigma,{\til\sigma}} {\mathit{f}(\til{x})}
            \lto{{y}!{v}}
            \lbinder y{\til \sigma}{\mathit{f}(\tilde{x})}
            }
        {
			\{\til x\} \subseteq \dom(\sigma) & f(\sigma(\til x)) \eval v
        }
        \\[1.5ex]
\qquad
            \infer[\m{AOutType}]
            {
            \lbinder y{\sigma,{\til\sigma}} {\mathit{f}(\til{x})}
            \lto{{y}!{B}}
            \lbinder y{\til \sigma}{\mathit{f}(\tilde{x})}
            }
        {
			\{\til x\} \subseteq \dom(\sigma) & f(\sigma(\til x)) \!\not\,\eval v
			& f : \_ \rightarrow B
        }        
\\[1.5ex]
\infer[\m{LInpNew}]
{
\lbinder y{\til\sigma}{f(\tilde{x})}
\lto{{x}?{\absv}}
\lbinder y{\til\sigma,\{x \mapsto \absv\}} {f(\tilde{x})}
}
{
x \in \bigcap_{\sigma_i \in \til \sigma}\dom(\sigma_i)
&
x \in \til x
}
\\[1.5ex]
	\infer[\m{LInpUpd}]
        {
        \lbinder y{\til\sigma_1,\sigma,\til\sigma_2} {f(\tilde{x})}
       \lto{{x}?{\absv}}
        \lbinder y{\til\sigma_1,\sigma[x \mapsto \absv],\til\sigma_2} {f(\tilde{x})}
        }
        {
		x \in \bigcap_{\sigma_i \in \til \sigma_1}\dom(\sigma_i)
		&
		x \not\in \dom(\sigma)
		&
		x \in \til x
        }
\\[1.5ex]
\infer[\m{LInpDisc}]
{
	\lbinder y{\til\sigma} {f(\tilde{x})}
       \lto{{x}?{\absv}}
    \lbinder y{\til\sigma} {f(\tilde{x})}
}
{
	x \not\in \til x
}
\\[1.5ex]
\infer[\m{LOutLift}]
       {
       L_1, L_2
       \lto{{y}!{\absv}}
       L_1', L_2
       }
       {
	L_1 \lto{{y}!{\absv}}L_1'}
\qquad
       \infer[\m{LInpList}]
       {
       L_1, L_2
       \lto{{x}?\absv}
       L_1', L_2'
       }
       {
	L_1 \lto{{x}?\absv}L_1'
\qquad
	L_2 \lto{{x}?\absv}L_2'}
	\\[1.5ex]
\end{array}
$$
\caption{Abstract Local Binder Semantics.}
\label{fig:abssem}
\end{figure*}

We may now introduce conformance between a set of abstract local binders and a local type,
defined by the rules shown in Fig.~\ref{fig:conformance}, described next. Essentially conformance simulates
the (abstract) behaviour of the local binders so as to assert its compatibility to the behaviour 
prescribed by the local type. Rule $\m{CInp}$ says
that $\absL$ is compatible with the input, if $\absL$ exhibits the (abstract) input, leading to 
a configuration that is conformant with the continuation of the input. Rule $\m{COut}$ follows
similar lines. Rule $\m{CBra}$ considers the two possible continuations, taking into account
the actual value that corresponds to each branch ($\inl$ or $\inr$). In such a way we may 
(statically) inform on the selection value of choices that depend on the branching. This notion is captured
in rules $\m{CChoL}$ and $\m{CChoR}$ where the selection value is available, hence the 
conformance may focus on the relevant branch. However, it is not always possible to statically
determine the selection value, hence the binders exhibit an output of $\chot$ type, for which
rule $\m{CCho}$ ensures that both branches are conformant with the type. 

Rules $\m{CRec}$
and $\m{CVar}$ ensure that the recursion body leaves the state unaltered, using a dedicated environment $\recenv$
(an association between recursion variables and abstract local binders). Intuitively, this means
that local binder queues are of the same size at the start of each
iteration of the recursion, hence that inputs are matched by respective outputs and vice-versa. Rule
$\m{CEnd}$ says that local binders are always conformant with inaction $\tend$, which reveals our
interpretation for conformance:
 binders can actually have more behaviour than the one prescribed by the local type, hence 
the focus is on ensuring local binders can carry out the specified behaviour. 
We remark that the simulation carried out for abstract local binders 
 strictly follows the type structure.

\begin{figure*}[t]
$$
\begin{array}{cc}
\infer[\m{CInp}]
{\recenv \vdash \absL \conforms \tinp{x}{B}{T} }
{
\absL \lto{{x}?{B}} \absL' \qquad \recenv \vdash \absL' \conforms T
}
&
\infer[\m{COut}]
{\recenv \vdash \absL \conforms  \tout{y}{B}{T} }
{
\absL \lto{{y}!{B}} \absL' \qquad \recenv \vdash \absL' \conforms T
}
\\[4mm]
\multicolumn{2}{c}{
\infer[\m{CBra}]
{\recenv \vdash \absL \conforms \tbranch{x}{T_1}{T_2}}
{
\absL \lto{{x}?{\inl}} \absL' \qquad \recenv \vdash \absL' \conforms T_1
\qquad
\absL \lto{{x}?{\inr}} \absL'' \qquad \recenv \vdash \absL'' \conforms T_2
}}
\\[4mm]
\multicolumn{2}{c}{
\infer[\m{CCho}]
{\recenv \vdash \absL \conforms \tchoice{y}{T_1}{T_2}}
{
\absL \lto{y!\chot} \absL' 
\qquad \recenv \vdash \absL' \conforms T_1
\qquad \recenv \vdash \absL' \conforms T_2
}
}
\\[4mm]
\infer[\m{CChoL}]
{\recenv \vdash \absL \conforms \tchoice{y}{T_1}{T_2}}
{
\absL \lto{y!\inl} \absL' 
\qquad \recenv \vdash \absL' \conforms T_1
}
\qquad&
\infer[\m{CChoR}]
{\recenv \vdash \absL \conforms \tchoice{y}{T_1}{T_2}}
{
\absL \lto{y!\inr} \absL' 
\qquad \recenv \vdash \absL' \conforms T_2
}
\\[4mm]
\multicolumn{2}{c}
{
\infer[\m{CRec}]
{\recenv \vdash \absL \conforms \mu \recvar X. T}
{\recenv, \recvar X : \absL \vdash \absL \conforms T }
\qquad
\infer[\m{CVar}]
{\recenv, \recvar X: \absL \vdash \absL \conforms \recvar X}
{}
\qquad
\infer[\m{CEnd}]
{\recenv \vdash \absL \conforms  \tend}
{}
}
\end{array}
$$
\caption{Conformance.}
\label{fig:conformance}
\end{figure*}

Having presented the main principles of our type system, we now introduce 
auxiliary syntactic notions used in the typing rules.
Namely, we require an operation that given local binders $L$ yields the abstract 
counterpart, denoted $\abs{L}$. The operation replaces all values different from $\inl$ and
$\inr$ by their respective types in the queues, leaving $\inl$ and $\inr$ untouched.
In order to realise the forwarder specification, we introduce an operation that 
renames a given type $T$ considering forwarders $F$, up to the sets of input and output ports, 
denoted $\rename{T}{F}{\tilde x}{\tilde y}$, defined in expected lines
(see Appendix~\ref{app:defs}).
We consider the set of ports so as to check that inputs
(and branches) specified in the type use input ports and that outputs (and choices) use output ports.
This operation thus ensures composite components respect the specified interfaces (otherwise it
is undefined), while the same principle is ensured for base components using a dedicated
predicate, denoted $\tilde x;\tilde y \vdash {T}$, defined in non-surprising lines.
%

We may now present our typing rules, shown in Fig.~\ref{fig:typing}. Rule $\m{TBaseC}$ addresses the 
base component, asserting it is well-typed provided the abstract version of the local binders is conformant
with the type of the component (considering an empty recursion environment). We also check that inputs
(and branches) are specified in the type considering input ports and symmetrically for output ports ($\tilde x; \tilde y \vdash T$). 
Rule $\m{TCompC}$ shows
the typing for the composite component. To type the subcomponent that interacts both with other subcomponents and with the external 
environment we 
consider a type merge: on the one hand the component participates in the protocol so the 
the projection of the protocol for the respective role is considered as part of the behaviour of 
the component. On the other hand the behaviour expected externally, up to a forwarder renaming ($\rename{T}{F}{\tilde x}{\tilde y}$), is also
carried out by the component. So the type merge combines the two separate behaviours
and yields the type expected for the component. For the remaining components we consider directly the projection of 
the protocol in the respective roles, where the association between message labels and base types $\labenv$
is the same so as to ensure communication safety. The set of roles considered is the set of roles
assigned in $R$ (noted $\mathit{roles}(R)$) so as to ensure all components are typed, potentially w.r.t. an $\tend$ 
projection since the set of roles used in the protocol may be smaller. This is the case for roles/components that 
have finished their contribution in the protocol, relevant for typing preservation.
Rules $\m{TRole}$ and $\m{TRoleL}$ distribute the typing assignments to the respective role assignments. 

\begin{figure*}[t]
$$
\begin{array}{c}
\infer[\m{TBaseC}]
{\chorboxb{\tilde x}{\tilde y}{ L} : T}
{
\emptyset \vdash \abs{ L} \conforms T 
\qquad
\tilde x; \tilde y \vdash T
}
\\\\
\infer[\m{TCompC}]
{
\chorbox{\tilde x}{\tilde y}{ D}{\role{r}[ F]}{ R}{G} : T
}
{
\begin{array}{c}
\projc{G}{\role r}{ D}{\labenv} \join 
\rename{T}{ F}{\tilde x}{\tilde y} = T_{\role r}
\qquad
\mathit{roles}(G) \subseteq \mathit{roles}(R) = \role r, \tilde {\role p} 
\\[1mm]
\role r : T_{\role r}, {\role p}_1: \projc{G}{{\role p}_1}{ D}{\labenv},\ldots,
{\role p}_k: \projc{G}{{\role p}_k}{ D}{\labenv} \vdash  R 
\end{array}
}
\\\\
\infer[\m{TRole}]
{\role r: T \vdash \roleas{r}{C}}
{C :T}
\qquad
\infer[\m{TRoleL}]
{\role r: T, {\role p}_1: T_1,\dots, {\role p}_k: T_k \vdash \roleas{r}{C},  R}
{ C: T \qquad {\role p}_1: T_1,\dots, {\role p}_k: T_k \vdash   R}
\end{array}
$$
\caption{Type System.}
\label{fig:typing}
\end{figure*}

Although an actual implementation of the type-checking procedure is out of 
the scope of this paper, we nevertheless remark that type-checking is decidable 
given all the required derivations (namely merge and conformance) are strictly 
bound by the structure of types.
Our type system ensure systems enjoy communication safety and progress
and informs on reuse and substitution principles.


We now present our main results and briefly discuss the proof structure.
We first capture soundness, by showing typing is preserved
under system evolution. 

\begin{theorem}[Typing Preservation]\label{the:subred}
Let $C:T$ and $C \lto{\lambda} C'$. Then:
\begin{itemize}
\item{if $\lambda = \tau$ then $C':T$.}
\item{if $\lambda \neq \tau$ and $T \lto{\lambda} T'$ then $C':T'$.}
\end{itemize}
\end{theorem}
\begin{proof}
By induction on the derivation of $C \lto{\lambda} C'$, building
on auxiliary results.
\end{proof}

Theorem~\ref{the:subred} identifies two separate cases: if the component 
performs an internal action, then the interface is untouched;
if the components exhibits an input or output and the local type matches
the communication action, then typing is preserved by the continuations.
Since GC semantics considers interaction is controlled by protocols, 
only behaviours prescribed by the protocol can actually occur.
We may then focus on actions available in the type as the protocol
projections capture the behaviour realisable by the respective
component. However, when typing the component that interfaces
with external and internal contexts
we must lift this notion up to type merge. 
Intuitively the component may exhibit actions that have been shuffled in the type structure 
(not immediately observable in the type), in which case we show that such 
actions can be shuffled out or \emph{swapped} and preserve typing. We remark 
that this property does not entail that any action can be shuffled out (otherwise 
merge would not be meaningful), as the proof regards actions that the component 
immediately exhibits (see Appendix~\ref{app:typeres}). 

Theorem~\ref{the:progress} states our our progress property,
attesting well-typed components can, after a number of internal steps, carry out the actions prescribed by the types.

\begin{theorem}[Progress]
\label{the:progress}
If $C:T$ and $T \lto{\lambda} T'$ then $C \ltow{\lambda} C'$.
\end{theorem}
\begin{proof}
By induction on the derivation of $C:T$ and on the derivation of the merge in the composite case, building on auxiliary results.
\end{proof}

The proof of Theorem~\ref{the:progress} crucially relies on an underlying progress principle for 
interaction controlled by protocols in the composite case. Actions expected by the 
external interface can only be ensured if all actions specified for the interfacing 
component can be carried out, namely actions which depend on the protocol which
in turn depend on other components. 
Corollary~\ref{cor:chorpro} captures protocol progress.

\setcounter{corollary}{2}

\begin{corollary}[Protocol Progress]
\label{cor:chorpro}
If $\chorbox{\tilde x}{\tilde y}{ D}{\role{r}[ F]}{ R}{G} : \tend$ and there is $G'$ such that 
$G \lto{\alpha} G'$ then
 $\chorbox{\tilde x}{\tilde y}{ D}{\role{r}[ F]}{ R}{G} \lto{\tau} 
 \chorbox{\tilde x}{\tilde y}{ D}{\role{r}[ F]}{ R'}{G'}$.
\end{corollary}

Corollary~\ref{cor:chorpro} thus attests that actions prescribed by the protocol are
carried out in well-typed components, focusing on components with a 
\emph{closed} interface ($\tend$). We thus have that protocols
not only control interaction among (sub)components but also provide a specification that is
guaranteed to be carried out. Thanks to the operational correspondence
result we may also ensure that well-typed components are compiled
to systems that enjoy communication safety and progress.

\subsection{Examples}

We return to the examples shown in Section~\ref{sec:examples} now considering a typing perspective.
Regarding Example A consider type 
$\tchoice{y_{l_1}}{\tinp {x_{l_2}} B \tend}{\tinp {x_{l_3}} B \tend}$ is
obtained by projecting the protocol in role $\role{p}$, where
we use the message labels as subscripts of port identifiers so as to yield the respective association and
assuming some base type $B$. Now consider the type
prescribed by the external environment is $\tbranch{x}{\tout{y_2}{B}{\tend}}{\tout{y_1}{B}{\tend}}$
(notice the association of $y_1$ and $y_2$ with the decision is ``inverted'' in the illustration given
previously). A possible merge is:
$$\tbranch{x}{\tchoice{y_{l_1}}{\tinp {x_{l_2}} B {\tout{y_2}{B}\tend}}{\tinp {x_{l_3}} B {\tout{y_2}{B}\tend}}}
{\tchoice{y_{l_1}}{\tinp {x_{l_2}} B {\tout{y_1}{B}\tend}}{\tinp {x_{l_3}} B {\tout{y_1}{B}\tend}}}
$$
which the component implementing role $p$ conforms to, thanks to the fact that the abstract local
binder semantics is able to capture the association between the branching on $x$ and the selection
on $y_{l_1}$ after which the dependencies between $x_{l_2}$ and $y_2$ and between $x_{l_3}$ 
and $y_1$ are met (notice this is so only in the relevant branches). 
We remark that if such 
external branching on $x$ is not present then the example would be untypable, since the component
must somehow inform on the alternative behaviours to the outer context. Such reasoning is also 
present in Example B given type $\tchoice{y_1}{\tout{y_2}{B}{\tend}}{\tout{y_1}{B}{\tend}}$
where the outer context is informed of the alternative behaviours via the selection on $y_1$.

Regarding Example C we remark that the choice of protocol (either $G_0$ or $G_1$) is constrained
by the type expected externally. If the external type is $\mu \recvar X. \tinp x B {\tout y B {\recvar X}}$
then necessarily (recursive) protocol $G_1$ must be used as otherwise the infinite behaviour
expected externally will not be realisable by the component. If the external type 
is $\tinp x B {\tout y B \tend}$ then necessarily protocol $G_0$ must be used since the internal
infinite dependencies are not met. The use of protocol $G_1$ in such circumstances would be
safe in a behavioural perspective but our current development does not capture this. Instead
in Example D the use of both protocols ($G_0$ and $G_1$) can be typed as there are no 
dependencies from external behaviour and internal communications.

In Example E the component is typable when considering the type prescribed by the environment
is $\tinp x B {\mu \recvar X. \tout y B {\recvar X}}$, given the ``one shot'' dependency between
receiving on $x$ and the internal communication $l_1$ is met, after which the (infinite) 
dependencies between the internal communication $l_2$ and the output on $y$ are also met.
Example F is typable considering external type $\mu \recvar X. \tbranch{x}{\recvar X}{\tout y B {\recvar X}}$,
but is untypable if we consider that the component is governed by protocol
$\mu\recvar X.\gcom{p}{l_1}{q}\ (\gcom{q}{l_2}{p};\tend,\
  \gcom{q}{l_3}{p}; 
  \recx X)$, a challenging pattern given the mixture of infinite
  and finite dependencies together with the choice propagation between external 
 and internal communications.


Lastly
we return to the example given in the language preview so as to illustrate further the notion
of possible (type safe) substitutions of the protocols specified in components. 
Namely,
consider $G_{\role{Shop}}$ defined as
$\gchoice{Sales}{buy}{Bank}
{\gcom{Sales}{val}{Bank};\gcom{Sales}{ccnum}{Bank}} {\tend}$. In the context
of the composite component implementing the $\role{Seller}$ we may
replace $G_{\role{Shop}}$ with $G_{\role{Shop}'}$ defined as  
$\gcom{Sales}{val}{Bank};\gchoice{Sales}{buy}{Bank}
{\gcom{Sales}{ccnum}{Bank}}{\tend}$ or 
$G_{\role{Shop}''}$ defined as
$\gchoice{Sales}{buy}{Bank}
{\gcom{Sales}{ccnum}{Bank}; \gcom{Sales}{val}{Bank}}{\tend}$,
where the $\mathit{val}$ message is swapped around, 
since the value of the product is determined (in $C_{\role{Sales}}$) 
as soon as the product name is available (originating from $\role{Buyer}$,
corresponding to the first message of the external protocol 
$G_{\text{BSS}}$). 

Also, considering the first composite component that uses $G_{\text{BSS}}$
and base components to implement the three roles, we may consider 
a recursive version of $G_{\text{BSS}}$ since the base components
are able to continuously react. However, such recursive $G_{\text{BSS}}$
cannot be used when considering the $\role{Seller}$ role is implemented
by the composite component, since its internal protocol $G_{\role{Shop}}$ 
does not specify a recursion, and hence the type merge is 
undefined (but a recursive version of $G_{\role{Shop}}$ would work).

\subsection{On Substitutability}

\begin{figure*}[t]
$$
\begin{array}{c}
\infer[\m{SIShuffle}]
{\tinp {x_1} B {\tinp {x_2} B T} \subt \tinp {x_2} B {\tinp {x_1} B T}}
{}
\qquad
\infer[\m{SOShuffle}]
{\tout {y_1} B {\tout {y_2} B T} \subt \tout {y_2} B {\tout {y_1} B T}}
{}
\\\\
%
\infer[\m{SODiscard}]
{\tout y B T \subt T}
{}
\qquad 
\infer[\m{SIBefore}]
{\tout {y} B {\tinp {x} B T} \subt \tinp {x} B {\tout {y} B T}}
{}
\\\\
\infer[\m{SO}]
{\tout y B {T_1} \subt \tout y B {T_2}}
{T_1 \subt T_2}
\qquad
\infer[\m{SI}]
{\tinp x B {T_1} \subt \tinp x B {T_2}}
{T_1 \subt T_2}
\qquad
\infer[\m{SRec}]
{\mu \recvar X. T_1 \subt \mu \recvar X. T_2}
{T_1 \subt T_2}
\\\\
\infer[\m{SCho}]
{\tchoice{y}{T_1}{T_2} \subt \tchoice{y}{T_1'}{T_2'}}
{T_1 \subt T_1' \quad T_2 \subt T_2'}
\qquad
\infer[\m{SBra}]
{\tbranch{x}{T_1}{T_2} \subt \tbranch{x}{T_1'}{T_2'}}
{T_1 \subt T_1' \quad T_2 \subt T_2'}
\\\\
\infer[\m{STra}]
{T_1 \subt T_2}{T_1 \subt T_3 \quad T_3 \subt T_2}
\qquad
\infer[\m{SRef}]
{T \subt T}{}
\qquad
\infer[\m{TSubsC}]
{C: T_2}
{C: T_1 \quad T_1 \subt T_2}
\end{array}
$$
\caption{Subtyping relation.}
\label{fig:subtyping}
\end{figure*}

We now present preliminary results and some insights regarding the substitution principle 
supported by our typed model. In particular we show a notion of subtyping that naturally 
arises in our setting which, combined with the standard subsumption rule, technically realises
the usual substitution principle (cf.~\cite{liskov94}): a component of type $T_1$ can replace a component of type $T_2$
when $T_1$ is a subtype of $T_2$ (denoted $T_1 \subt T_2$). We also discuss further extensions
of the subtyping relation so as to encompass further generality and hint on the underlying semantics.

Subtyping is given by the rules in Fig.~\ref{fig:subtyping}, including type language closure, transitivity, and reflexivity and four other rules described next under the light of the substitution principle. 
Rules $\m{SIShuffle}$ and
$\m{SOShuffle}$ say that two inputs and two outputs can be safely shuffled between them, and
rule $\m{SODiscard}$ says a component may be safely used in a context where one of the 
outputs provided by the component is not expected. Notice the queues in local binders can store
(unused) data, and receive two inputs or exhibit two outputs in any order. Instead,
rule $\m{SIBefore}$ addresses shuffling of one output and one input: when the component
is able to first output and then input then it is safe to use it in a context that prescribes that the input
happens first. We remark that the symmetric does not hold, in particular when the output 
\emph{depends} on the input. We may show type safety considering the extension with the
subsumption rule for components in Fig.~\ref{fig:subtyping} ($\m{TSubsC}$)
and the presented subtyping relation, 
and we can already identify extensions to the subtyping relation.

As hinted above, the shuffling of outputs and inputs does not hold in general, in particular 
when the component actually requires a received value in order to be able to exhibit some
output. So considering the rule symmetric to $\m{SIBefore}$ (dubbed $\m{SOBefore}$)
is unsound in our setting, however an extension of our typing system so as capture input/output 
value dependencies would allow the precise application of $\m{SOBefore}$. The extension
would collect direct dependencies in local binders of base components and inductively 
extend them by adding the dependencies introduced by protocols in composite components.
Such information would also support a rule for discarding inputs (analogous to $\m{SODiscard}$)
when no (relevant) outputs depend on the (discarded) input. Also, extending the typing information so
as to include information on the ports in which the component is actually able to interact, regardless
of being mentioned or not initially in the type, would allow to reason on subtyping principles such as
${T \subt \tinp x B T}$ and ${T \subt \tout y B T}$ which are sound up to such ``unused'' capabilites 
of the component (and when the dependencies are met in the case of the output). Resorting to the 
used/unused information on ports and the mentioned dependencies we may also reason on 
principles such as $\mu \recvar X. T$ is a subtype of $T$ and $T$ is a subtype of $\tend$, which are not
sound in our current setting since the underlying dependencies may be impossible to satisfy. 

Finally, we remark that we may synthesise components that may be used in a context when a 
given type is expected, for any given type. In particular given type $T$ we may devise
a base component that defines input ports for all inputs and branches specified in $T$
(not used in internal local binders), and output ports for all outputs
and choices (defined with functions of the appropriate type and without parameters).
We expect to find such components at the bottom of the hierarchy when considering
a semantic notion for substitutability, going up to components of closed interfaces 
at the top for which adaptation to surrounding contexts is limited.



%

\section{Concluding Remarks}
The protocol language used in this
article sits in between the two related approaches of Multiparty
Session Types (MPST)~\cite{HYC16} and Choreographic
Programming~\cite{M13:phd}. MPST are types for specifying distributed
protocols that can be used for checking whether a given distributed
implementation respects the given protocol. The methodology approach
for MPST is different from ours: protocols are also given as
choreographies, but they are only used statically. In our setting,
choreographies are used at runtime for governing the way components,
our data handlers, communicate with each other. MPST have also been
used as monitors during runtime execution of concurrent
programs~\cite{CBDHY11,BCDHY17}. However, rather than monitoring
reactive components like in our work, they monitor $\pi$-calculus-like
processes, hindering modularity.

Choreographic Programming~\cite{CM13} allows to program distributed
systems directly as choreographies which define both the
communications among components and the data that they carry.  Then,
the choreography is used to synthesise a correct process
implementation.  Differently, our language decouples the component
behaviour where choreographies are used to synthesise
correct-by-construction process implementations.  Similarly, the I/O
actions performed by components in this work are synthesised from the
choreographies that they participate in.  Also in this case, the
absence of decoupling component interaction from component behaviour
limits modularity.
 
Reo~\cite{A04} is a coordination language which separates protocols
and implementation of components. However, their model for
implementing processes still requires to provide sequences of
communication actions, making reusability (hence modularity)
harder. Instead, in our language, thanks to the usage of reactive
programming, we can be more flexible with the type of components that
can be used. 
MECo~\cite{MS10} is a calculus where components interact through their
input and output ports, similar to our interfaces.  MECo supports
notions of passivation and component mobility, which could be an
interesting future extensions for our language.

Orc~\cite{MC07} is an orchestration language that uses connectives
between components similar to our binders. However, the language has
no choreographic coordination on the message flow. Similarly,
BIP~\cite{bippaper} is a coordination language for modularly
assembling components. Also in BIP, there is no usage of choreographic
specification for governing the interaction between components.

We are yet to fully explore the relation with existing literature addressing 
reactive programming and related models (e.g.,~\cite{alur99,alfaro01,ostrowski08,talpin97}),
which we believe may be used so as to enrich our setting. We have already identified some 
connection points, in particular at the level of the subtyping notion~\cite{ostrowski08}.
It would also be interesting to enrich our setting by exploring the relation between linear
logic and session types~\cite{CP10} and also integrating with dependent types so as
to, for instance, capture choice propagation in a more refined way (cf.~\cite{toninhofossacs18}).

We have presented GC, a language designed for the
modular development of distributed systems, where reactive components
are assembled in a choreographic governed way. We show a
provably-correct operational interpretation of our model, which
shows a distributed implementation of the governing carried out by protocols, 
and a type discipline that ensures
communication safety and progress, already supporting a notion of substitutability.  
Future directions for this work
include an actual implementation to serve as an experimental proof of
concept. Also, at the level of typing, it would be interesting to
consider an interpretation of component types that capture their
behaviour in the most general way. In particular this is crucial to obtain a more precise notion 
of subtyping so as to 
provide further support for reuse and substitution, thus a
fundamental research direction for this work.

%

\newpage



%
%
\bibliography{biblio}
%
\newpage

\appendix

\section{Auxiliary Definitions}
\label{app:defs}

Given choreography $G$, local binders $L$ and mapping $\gamma$ we define 
$\queues{G}{L}{\mapping}$ for a given by induction 
on the structure of the choreography, as shown in Fig.~\ref{fig:chorqs}.
For the case of 
(in transit) communication, the queues for receivers waiting to receive the 
message, 
identified via the mapping ($\mapping(\role{q}_i, \lab)$), are augmented 
considering
the value registered in the choreography and the input port $u'$. 
The operation is defined only for local binders defined by $\buildqueues{}{}$,
hence by construction the port $u'$ is the same (originally given 
by $\mapping(\role p, \lab)$ for sender role $\role{p}$). The 
operation proceeds to the continuation considering the augmented queues,
thus ensuring message order is preserved. The difference in the 
cases for (in transit) choice is that the the focus is on the branch
implied by the registered value.

\begin{figure}[t]
\begin{displaymath}
\begin{array}{lclll}
\queues {\gcom p\lab{ \tilde q};G}{L}{\mapping}
 \deff 
\queues {G}{L}{\mapping}
\\
\queues {{\gcom {}{\lab,v}{ \tilde q};G}}
{L, \lbinder {u_1} {\tilde \sigma_1} {\mathit{id}(u')},
\ldots, \lbinder {u_k} {\tilde \sigma_k} {\mathit{id}(u')}}
{\mapping}
\\ \hfill\deff 
\queues {G}{
{L, \lbinder {u_1} {\tilde \sigma_1, \{u' \mapsto v\}} {\mathit{id}(u')},
\ldots, \lbinder {u_k} {\tilde \sigma_k, \{u' \mapsto v\}} {\mathit{id}(u')}}
}{\mapping}
\\
\hfill
(u_i = \mapping(\role q_i, \lab) \wedge \tilde q =  q_1, \ldots, q_k)
\\
\queues{\gchoice p\lab{\tilde q}{G_1}{G_2})}{L}{\mapping}  \deff 
\queues {G_1}{L}{\mapping}  
\hfill  (\queues {G_1}{L}{\mapping} = \queues {G_2}{L}{\mapping})
\\
\queues{\gchoice {}{\lab,\inl}{\tilde q}{G_1}{G_2}}
{L, \lbinder {u_1} {\tilde \sigma_1} {\mathit{id}(u')},
\ldots, \lbinder {u_k} {\tilde \sigma_k} {\mathit{id}(u')}}
{\mapping}
\\ \hspace{1cm}\deff 
\queues{G_1}
{L, \lbinder {u_1} {\tilde \sigma_1, \{u' \mapsto \inl\}} {\mathit{id}(u')},
\ldots, \lbinder {u_k} {\tilde \sigma_k, \{u' \mapsto \inl\}} {\mathit{id}(u')}}
{\mapping}
\\
\hfill
(u_i = \mapping(\role q_i, \lab) \wedge \tilde q =  q_1, \ldots, q_k)
\\
\queues{\gchoice {}{\lab,\inr}{\tilde q}{G_1}{G_2}}
{L, \lbinder {u_1} {\tilde \sigma_1} {\mathit{id}(u')},
\ldots, \lbinder {u_k} {\tilde \sigma_k} {\mathit{id}(u')}}
{\mapping}
\\ \hfill \deff
\queues{G_2}
{L, \lbinder {u_1} {\tilde \sigma_1, \{u' \mapsto \inr\}} {\mathit{id}(u')},
\ldots, \lbinder {u_k} {\tilde \sigma_k, \{u' \mapsto \inr\}} {\mathit{id}(u')}}
{\mapping}
\\
\hfill
(u_i = \mapping(\role q_i, \lab) \wedge \tilde q =  q_1, \ldots, q_k)
\\[1mm]
\queues{ \mu \recvar X.G}{L}{\mapping} \deff  \queues{G}{L}{\mapping}
\qquad
\queues{ \recvar X}{L}{\mapping} \deff L
\qquad
\queues{ \tend}{L}{\mapping} \deff L
\end{array}
\end{displaymath}
\caption{Fill Operation.}
\label{fig:chorqs}
\end{figure}

Given type $T$ and ports $\tilde x$ and $\tilde y$, forwarder renaming $\rename{T}{F}{\tilde x}{\tilde y}$ 
is given by the rules shown in Fig.~\ref{fig:frename} and interface check $\tilde x;\tilde y \vdash {T}$ 
is defined by the rules shown in Fig.~\ref{fig:interface}. 

\begin{figure*}
\begin{displaymath}
\begin{array}{c}
\infer[\m{FInp}]
{\rename{ \tinp{x_i}{B}{T}}{F}{\tilde x}{\tilde y} = \tinp{z}{B}{T'} }
{\rename{ {T}}{F}{\tilde x}{\tilde y} = T' \qquad x_i \in \tilde x \qquad F = \fbinder z{x_i}, F' }
\\[4mm]
\infer[\m{FOut}]
{\rename{ \tout{y_i}{B}{T}}{F}{\tilde x}{\tilde y} = \tout{z}{B}{T'} }
{\rename{ {T}}{F}{\tilde x}{\tilde y} = T' \qquad y_i \in \tilde y \qquad F = \fbinder {y_i}z, F' }
\\[4mm]
{
\infer[\m{FBranch}]
{\rename{ \tbranch{x_i}{T_1}{T_2}}{F}{\tilde x}{\tilde y} = \tbranch{z}{T_1'}{T_2'}}
{\rename{ {T_1}}{F}{\tilde x}{\tilde y} = T_1' \qquad \rename{ {T_2}}{F}{\tilde x}{\tilde y} = T_2' 
\qquad x_i \in \tilde x \qquad F = \fbinder z{x_i}, F' }
}
\\[4mm]
{
\infer[\m{FChoice}]
{\rename{ \tchoice{y_i}{T_1}{T_2}}{F}{\tilde x}{\tilde y} = \tchoice{z}{T_1'}{T_2'}}
{\rename{ {T_1}}{F}{\tilde x}{\tilde y} = T_1' \qquad \rename{ {T_2}}{F}{\tilde x}{\tilde y} = T_2'
 \qquad y_i \in \tilde y \qquad F = \fbinder {y_i}z, F' }
 }
\\[4mm]
{
\infer[\m{FRec}]
{\rename{ \mu \recvar X. T}{F}{\tilde x}{\tilde y} = \mu \recvar X. T'}
{\rename{ T}{F}{\tilde x}{\tilde y} = T'}
\qquad
\infer[\m{FVar}]
{\rename{ \recvar X}{F}{\tilde x}{\tilde y} = \recvar X}
{}
\qquad
\infer[\m{FEnd}]
{\rename{\tend}{F}{\tilde x}{\tilde y} }
{}
}
\end{array}
\end{displaymath}
\caption{Forwarder Renaming}
\label{fig:frename}
\end{figure*}

\begin{figure*}
\begin{displaymath}
\begin{array}{cc}
\infer[\m{IInp}]
{\tilde x;\tilde y \vdash \tinp{z}{B}{T} }
{
\tilde x;\tilde y \vdash {T} \qquad z \in \tilde x
}
&
\infer[\m{IOut}]
{\tilde x;\tilde y \vdash  \tout{z}{B}{T} }
{
\tilde x;\tilde y \vdash T \qquad z \in \tilde y
}
\\[4mm]
\infer[\m{IBra}]
{\tilde x;\tilde y \vdash \tbranch{z}{T_1}{T_2}}
{
\tilde x;\tilde y \vdash T_1
\qquad
\tilde x;\tilde y \vdash T_2
\qquad
z \in \tilde x
}
\qquad &
\infer[\m{ICho}]
{\tilde x;\tilde y \vdash \tchoice{z}{T_1}{T_2}}
{
\tilde x;\tilde y \vdash T_1
\qquad
\tilde x;\tilde y \vdash T_2
\qquad
z \in \tilde y
}
\\[4mm]
\multicolumn{2}{c}
{
\infer[\m{IRec}]
{\tilde x;\tilde y \vdash \mu \recvar X. T}
{\tilde x;\tilde y \vdash T }
\qquad
\infer[\m{IVar}]
{\tilde x;\tilde y \vdash \recvar X}
{}
\qquad
\infer[\m{IEnd}]
{\tilde x;\tilde y \vdash \tend}
{}
}
\end{array}
\end{displaymath}
\caption{Interface}
\label{fig:interface}
\end{figure*}

\section{Auxiliary Results for Encoding}
\label{app:encres}

\begin{lemma}
\label{lem:compout}
We have that:
\begin{itemize}
\item{If $C \lto{y!v} C'$ then $\projb{C} \lto{y!v} \equiv \projb{C'}$.}
\item{If $C \lto{x?v} C'$ then $\projb{C} \lto{x?v} \equiv \projb{C'}$.}
\end{itemize}
\end{lemma}
\begin{proof}
By induction on the derivation of $C \lto{y!v} C'$ and of $C \lto{x?v} C'$.
\end{proof}

\begin{lemma}
\label{lem:chorout}
We have that:
\begin{itemize}
\item{If $G \lto{\labout{\role p}{\lab}{v}} G'$ then 
$\projb{G}_{\role p, D, \mapping} 
\lto{z'?v} P \lto{u!v} \equiv \projb{G'}_{\role p, D, \mapping}$,
given $D = \dbinder {\lab}{ q}{ w}{\roleport {p}{z}}, D' $
and $\mapping(\role{p}, z) = z'$ and $\mapping(\role{p}, \lab) = u$.}
\item{If $G \lto{\labinp{\role q}{\lab}{v}} G'$ then $\projb{G}_{\role q, D, \mapping} 
\lto{u?v} P \lto{w'!v} \equiv \projb{G'}_{\role q, D, \mapping}$
given 
$D = \dbinder {\lab}{ q}{ w}{\roleport {p}{z}}, D' $
and $\mapping(\role{q}, w) = w'$ and $\mapping(\role{q}, \lab) = u$.}
\end{itemize}
\end{lemma}
\begin{proof}
By induction on the derivation of $G \lto{\labout{\role p}{\lab}{v}} G'$
and of $G \lto{\labinp{\role q}{\lab}{v}} G'$.
\end{proof}

\section{Auxiliary Results for Typing}
\label{app:typeres}

\begin{lemma}[Swap Composite]
\label{lem:swapcomp}
If $C:T$ and $T = T_1 \join T_2$ and $C \lto{\lambda} C'$ 
and $T_1 \lto{\lambda} T_1'$ then there is $T'$ such that 
$T' = T_1' \join T_2$ and $C' : T'$.
\end{lemma}
\begin{proof}
By induction on the derivation of $T = T_1 \join T_2$, using Lemmas~\ref{lem:swapbase} and~\ref{lem:mergeprops}.
\end{proof}


\begin{lemma}[Swap Base]
\label{lem:swapbase}
If $\emptyset \vdash \abs{ L} \conforms T$ and 
$ L \lto{\lambda}  L'$
and $T = T_1 \join T_2$ and $T_1 \lto{\lambda} T_1'$ then
there is $T'$ such that $T' = T_1' \join T_2$ and 
$\emptyset \vdash \abs{ L'} \conforms T'$.
\end{lemma}
\begin{proof}
By induction on the derivation of $T = T_1 \join T_2$, using Lemmas~\ref{lem:harmony} and~\ref{lem:diamond}.
\end{proof}


\begin{lemma}[Swap Abstract]
\label{lem:diamond}
If $ \absL \lto{\lambda_1}  \absL_1$ and $ \absL \lto{\lambda_2}  \absL_2$ then
there is $ \absL_3$ such that $ \absL_1 \lto{\lambda_2}  \absL_3$ and
$ \absL_2 \lto{\lambda_1}  \absL_3$.
\end{lemma}
\begin{proof}
The proof follows directly from the semantics of abstract local binders.
\end{proof}


\begin{lemma}[Harmony of the Abstract Semantics]
\label{lem:harmony}
\begin{itemize}
\item{
$ L \lto{\lambda}  L'$ iff $\abs{ L} \lto{\abs{\lambda}} \abs{ L'}$ when $ y!\inl \neq \lambda \neq  y!\inr$.}

\item{
$ L \lto{y!\inl}  L'$
iff ($\abs{ L} \lto{y!\inl} \abs{ L'}$ or $\abs{ L} \lto{y!\chot} \abs{ L'}$).
}
\item{If
$ L \lto{y!\inr}  L'$ iff ($\abs{ L} \lto{y!\inr} \abs{ L'}$ or $\abs{ L} \lto{y!\chot} \abs{ L'}$).
}
\end{itemize}
\end{lemma}
\begin{proof}
The proof follows directly by the definition of $\abs{}$.
\end{proof}


\begin{lemma}[Merge is Associative and Commutative]
If $T = T_1 \join T_2$:
\label{lem:mergeprops}
\begin{itemize}
\item{
and $T_2 = T_2' \join T_2''$ then
there is $T_3$ such that $T_3 = T_1 \join T_2'$ and
$T = T_3 \join T_2''$.}
\item{then $T = T_2 \join T_1$.}
\end{itemize}
\end{lemma}
\begin{proof}
By induction on the derivation of $T = T_1 \join T_2$ and $T_2 = T_2' \join T_2''$.
\end{proof}

\end{document}